\documentclass{sig-alternate-05-2015}

\newfont{\mycrnotice}{ptmr8t at 7pt}
\newfont{\myconfname}{ptmri8t at 7pt}
\let\crnotice\mycrnotice%
\let\confname\myconfname%
\usepackage{amsmath,amssymb,latexsym}
\usepackage{paralist}
\usepackage{array}
\usepackage{cancel}
\usepackage{color}
\usepackage{xspace}
\usepackage{subcaption}
\usepackage{tikz}

\usetikzlibrary{decorations}
\usetikzlibrary{decorations.pathmorphing}
\usetikzlibrary{decorations.pathreplacing}
\usetikzlibrary{decorations.shapes}
\usetikzlibrary{decorations.text}
\usetikzlibrary{decorations.markings}
\usetikzlibrary{decorations.fractals}
\usetikzlibrary{decorations.footprints}
\usetikzlibrary{shapes}

\newcommand{\set}[1]{\left\{ #1 \right\}}
\newcommand{\card}[1]{\left| #1 \right|}
\newcommand{\varineq}{\vartriangleleft}
\newcommand{\baserel}{A_{\sf Bse}}
\newcommand{\targetrel}{A_{\sf Sol}}
\newcommand{\targetcard}{M_{\sf Sol}}
\newcommand{\ur}{\mathit{UR}}
\newcommand{\I}{\mathcal{I}}

\newcommand{\bod}{\mathsf{BoD}}
\newcommand{\bode}{\mathsf{BoD_E}}
\newcommand{\bodu}{\mathsf{BoD_U}}
\newcommand{\sod}{\mathsf{SoD}}
\newcommand{\sode}{\mathsf{SoD_E}}
\newcommand{\sodu}{\mathsf{SoD_U}}
\newcommand{\gbl}{\mathsf{G}_{card}}
\newcommand{\lcl}{\mathsf{L}_{card}}
\newcommand{\fnc}{(=, 1)}

\newcommand{\APEP}{\textsc{APEP}\xspace}
\newcommand{\DAPEP}{\textsc{D-APEP}\xspace}

\newcommand{\OAPEP}{\textsc{O-APEP}\xspace}
\newcommand{\maxAPEP}{\textsc{maxAPEP}\xspace}

\newcommand{\APEParg}[1]{\APEP\langle #1 \rangle}
\newcommand{\DAPEParg}[1]{\DAPEP\langle #1 \rangle}

\newcommand{\maxAPEParg}[1]{\maxAPEP\langle #1 \rangle}

\newtheorem{theorem}{Theorem}

\newtheorem{corollary}[theorem]{Corollary}

\newtheorem{definition}[theorem]{Definition}

\newtheorem{lemma}[theorem]{Lemma}

\newtheorem{proposition}[theorem]{Proposition}
\newtheorem{remark}[theorem]{Remark}

\begin{document}

\title{The Authorization Policy Existence Problem}

\numberofauthors{4}
\author{
\alignauthor
Pierre Berg\'e\\
  \affaddr{LRI, Universit\'e Paris-Saclay}\\
\affaddr{B\^{a}t 650, Rue Noetzlin, 91190 Gif-sur-Yvette}\\
  \affaddr{France}
  \email{Pierre.Berge@supelec.fr}
\alignauthor
Jason Crampton\\
\affaddr{Royal Holloway}\\
 \affaddr{University of London}\\
\affaddr{Egham, TW20 9QY} \\
\affaddr{United Kingdom} 
 \email{jason.crampton@rhul.ac.uk}\\
\and
\alignauthor
Gregory Gutin \\
\affaddr{Royal Holloway}\\
\affaddr{University of London}\\
\affaddr{Egham, TW20 9QY} \\
\affaddr{United Kingdom} 
 \email{g.gutin@rhul.ac.uk}
\alignauthor
R\'emi Watrigant \\
  \affaddr{INRIA Sophia-Antipolis}\\
  \affaddr{2004 route des Lucioles}\\
  \affaddr{06902 Sophia-Antipolis, France}
   \email{remi.watrigant@inria.fr}
}

% \date{}

\CopyrightYear{xxx}
\setcopyright{acmcopyright}
\conferenceinfo{xxx}{xxx}
\isbn{xxx}\acmPrice{xxx}
\doi{xxx}

\maketitle%
\begin{abstract}
Constraints such as separation-of-duty are widely used to specify requirements that supplement basic authorization policies.
However, the existence of constraints (and authorization policies) may mean that a user is unable to fulfill her/his organizational duties because access to resources has been denied.
In short, there is a tension between the need to protect resources (using policies and constraints) and the availability of resources.
Recent work on workflow satisfiability and resiliency in access control asks whether this tension compromises the ability of an organization to achieve its objectives.
In this paper, we develop a new method of specifying constraints which subsumes much related work and allows a wider range of constraints to be specified.
The use of such constraints leads naturally to a range of questions related to ``policy existence'', where a positive answer means that an organization's objectives can be realized.
We analyze the complexity of these policy existence questions and, for particular sub-classes of constraints defined by our language, develop fixed-parameter tractable algorithms to solve them.
\end{abstract}

\begin{CCSXML}
<ccs2012>
<concept>
<concept_id>10002978.10002991.10002993</concept_id>
<concept_desc>Security and privacy~Access control</concept_desc>
<concept_significance>500</concept_significance>
</concept>
<concept>
<concept_id>10002978.10002986.10002988</concept_id>
<concept_desc>Security and privacy~Security requirements</concept_desc>
<concept_significance>300</concept_significance>
</concept>
<concept>
<concept_id>10003752.10003809.10010052.10010053</concept_id>
<concept_desc>Theory of computation~Fixed parameter tractability</concept_desc>
<concept_significance>500</concept_significance>
</concept>
</ccs2012>
\end{CCSXML}

\ccsdesc[500]{Security and privacy~Access control}
\ccsdesc[300]{Security and privacy~Security requirements}
\ccsdesc[500]{Theory of computation~Fixed parameter tractability}

\printccsdesc

\keywords{access control; resiliency; satisfiability; computational complexity; fixed-parameter tractability}

\vfill\eject

\section{Introduction}\label{sec:intro}

Access control is a fundamental aspect of the security of any multi-user computing system, and is typically based on the specification and enforcement of an authorization policy.
Such a policy identifies which interactions between users and resources are to be allowed by the system.

Over the last twenty years, access control requirements have become increasingly complex, leading to increasingly sophisticated authorization policies, often expressed in terms of constraints.
A separation-of-duty constraint (also known as the ``two-man rule'' or ``four-eyes policy'') may, for example, require that no single user is authorized for some particularly sensitive group of resources.
Such a constraint is typically used to prevent misuse of the system by a single user.

The use of authorization policies and constraints, by design, limits which users may access resources.
Nevertheless, the ability to perform one's duties requires access to particular resources, and overly prescriptive policies and constraints may mean that some resources are inaccessible.
In short, ``tension'' may exist between authorization policies and operational demands: too lax a policy may suit organizational demands but lead to security violations; whereas too restrictive a policy may compromise an organization's ability to meet its business objectives.

Recent work on workflow satisfiability and access control resiliency has recognized the importance of being able to identify whether or not security policies prevent an organization from achieving its objectives~\cite{CrGuWa15,CrGuYe13,LiWaTr09,MaMoMo14,WaLi10}.
In this paper, we seek to generalize existing work in this area.
Specifically, we introduce the {\sc Authorization Policy Existence Problem} (\APEP), which may be treated as a decision or optimization problem.
Informally, \APEP seeks to find an authorization policy, subject to restrictions on individual authorizations (defined by a ``base'' authorization relation) and restrictions on collective authorizations (defined by a set of authorization constraints).
We show that a number of problems in the literature on workflow satisfiability and resiliency are special cases of \APEP, thereby showing that \APEP is computationally hard.

The framework within which \APEP is defined admits a greater variety of constraints than is usually considered in either the standard access control literature~\cite{BrNa89,GlGaFe98,LiTrBi07,SaCoFeYo96} or in workflow satisfiability~\cite{BaBuKa14,CoCrGaGuJo14,WaLi10}.
In this paper we characterize the constraints of interest and extend the definition of user-independent constraints~\cite{CoCrGaGuJo14} to this framework.
We then establish the complexity of \APEP for certain types of constraints, using both classical and multi-variate complexity analysis.
In this paper, we make some progress in this direction by establishing the complexity of \APEP for the constraints that we believe will be the most useful in practice.
In particular, we establish connections between \APEP and both the {\sc Workflow Satisfiability Problem} and resiliency in access control.

In the next section, we summarize relevant background material and related work.
We introduce the \APEP in Section~\ref{sec:apep} and elaborate on the nature of the constraints we consider in Section~\ref{sec:constraints}.
In Section~\ref{sec:apep-complexity}, we investigate the complexity of several variants of the \APEP.
We then discuss further constraint types to and connections between \APEP and existing problems in the literature.
We conclude the paper with a summary of our contributions and some ideas for future work.

\section{Background and related work}\label{sec:background}

A number of interesting (and computationally hard) problems arise naturally in the context of authorization policies and constraints.
However, the relative sizes of the parameters in many of these problems mean that it is fruitful to analyze these problems using multivariate complexity analysis.
In this section, we review some of those problems and provide a brief introduction to fixed-parameter tractability.

\subsection{Fixed-parameter tractability}

Many problems take multiple inputs and the complexity of solving such problems is determined by the sizes of those inputs.
In general, a problem may be hard in terms of the total size of the input.
However, if we consider the complexity of a problem under an assumption that some of the parameters of the input are small and terms that are exponential in those parameters are acceptable, then we may discover that relatively efficient algorithms exist to solve the problem.

\sloppy
More formally, an algorithm is said to be \emph{fixed-parameter tractable} (FPT) if it solves a decision problem in time $O(f(k)p(n))$, where $k$ is some (small) parameter of the problem, $n$ is the total size of the input, and $f$ and $p$ are, respectively, a computable function and a polynomial. 
As is customary in the literature on FPT algorithms, we will often write $O(f(k)p(n))$ as $O^*(f(k))$.
(That is, $O^*$ suppresses polynomial factors, as well as multiplicative constants.)
If a problem can be solved using an FPT algorithm then we say that it is an \emph{FPT problem} and that it belongs to the class FPT~\cite{DowFel13,Ni06}. 
An FPT algorithm for a hard problem is particularly valuable when $k$ is significantly smaller than $n$ for most instances of the problem that arise in practice.
% Naturally, the term $f(k)$ will be exponential, but
% The class FPT significantly extends the class of computationally tractable problems, i.e. problems that can be solved efficiently. 
In particular, FPT algorithms of practical value have been developed for the {\sc Workflow Satisfiability Problem} and its generalizations~\cite{CoCrGaGuJo14,CrGuKa15,KaGaGu15}.

\fussy

\subsection{Workflow satisfiability}

A workflow may be modeled as a set of steps in some automated business process.
An authorization relation determines which users are authorized to perform which steps, and constraints restrict which subsets of users may perform subsets of steps~\cite{BeFeAt99,Cr05,WaLi10}.
Given a set of users $U$, a set of workflow steps $S$, an authorization relation $A \subseteq S \times U$, and a set of constraints $C$, a plan $\pi : S \rightarrow U$ is a function allocating users to steps.
A valid plan must allocate an authorized user to each step and every constraint must be satisfied.
The {\sc Workflow Satisfiability Problem} (WSP) asks whether there exists a valid plan for a given $U$, $S$, $A$ and $C$.
Basin, Burri and Karjoth model a workflow (with ``break points'') as a process algebra and introduced the notion of an \emph{enforcement process}, which is analogous to a valid plan~\cite{BaBuKa14}.
This leads naturally to the \emph{enforcement process existence} problem, which is analogous to the workflow satisfiability problem and inspires the name for the problem we study in this paper.

Wang and Li observed that fixed-parameter algorithmics is an appropriate way to study WSP, because the number of steps is usually small and often much smaller than the number of users.%
 \footnote{The SMV loan origination workflow studied by Schaad \emph{et al.}, for example, has 13 steps and identifies five roles~\cite{ScLoSo06}. 
	   It is generally assumed that the number of users is significantly greater than the number of roles.}
Wang and Li \cite{WaLi10}  proved that WSP is FPT if we consider only separation-of-duty and binding-of-duty constraints \cite{WaLi10}. We will denote such constraints as $(s,s',\neq)$ and $(s,s',=)$, respectively, where $s,s'\in S$. A plan $\pi$ satisfies a constraint  $(s,s',\neq)$ ($(s,s',=)$, respectively) if $\pi(s)\neq \pi(s')$ ($ \pi(s) = \pi(s')$, respectively). WSP with only separation-of-duty constraints (only binding-of-duty constraints, respectively) will be denoted by WSP($\neq$) (WSP(=), respectively). 
Subsequent research has shown that WSP remains FPT even if additional types of constraints, notably user-independent~\cite{CoCrGaGuJo14} and class-independent constraints~\cite{CrGaGuJo15}, are permitted in the input to WSP. 
Note that WSP is not FPT in general \cite{WaLi10} unless a widely-accepted hypothesis in complexity theory fails; a significant body of research suggests this is highly unlikely~\cite{DowFel13}. 

\subsection{Static separation-of-duty constraints}

Constraints have been studied extensively in the context of role-based access control (RBAC)~\cite{GlGaFe98,LiTrBi07,SaCoFeYo96,SiZu97}.
In its simplest form, a static separation-of-duty constraint may be defined as a pair of roles $\set{r,r'}$ belonging to the set of roles $R$.
A user-role assignment relation $\ur \subseteq U \times R$, where $U$ is the set of users, satisfies the constraint $\set{r,r'}$ if there is no user $u$ such that $(u,r)$ and $(u,r')$ belong to $\ur$.

More generally, Li, Tripunitara and Bizri \cite{LiTrBi07} introduced the notion of a $q$-out-of-$m$ static separation-of-duty constraint, defined as a pair $(R',m)$, where $R'$ is a subset of $R$ of cardinality $q$.
A user-role assignment relation $\ur \subseteq U \times R$ satisfies the constraint $(R',m)$ if there is no set of $t < q$ users that are collectively authorized for $R'$.
That is, for all subsets $V$ of $U$ such that $\card{V} < q$,
\[
 \bigcup_{u \in V} \set{r : (u,r) \in \ur} \subset R' .
\]
Note that the simple separation-of-duty constraint defined by a pair of roles $\set{r,r'}$ is a $2$-out-of-$2$ separation-of-duty constraint.

\subsection{Resiliency}

Li, Wang and Tripunitara introduced the idea of resiliency in access control~\cite{LiWaTr09}.
Informally, a resiliency policy is a requirement that even in absence of a limited number of users the remaining users can be authorized for some set of resources such that given constraints are satisfied.
%an authorization relation enables multiple teams to be formed, such that the users in each team are collectively authorized for some set of resources.
%Thus, a resiliency policy is a requirement on users to be authorized, whereas a separation-of-duty policy is a requirement restricting which users are authorized.
The existence of both a resiliency policy  and authorization constraints may mean that no authorization relation can satisfy all requirements.
Li \emph{et al.} introduce a number of problems investigating whether an authorization relation does exist~\cite{LiWaTr09}.

\section{The authorization policy\\ existence problem}\label{sec:apep}

In this paper, we extend existing work on workflow satisfiability, constraints and resiliency, by defining a simple yet very expressive authorization framework. Roughly speaking, we specify a problem dealing with the existence of an appropriate authorization relation.

Given a set of users $U$ and a set of resources $R$ to which access should be restricted, we may define an \emph{authorization relation} $A \subseteq U \times R$, where $(u,r) \in A$ if and only if $u$ is authorized to access $r$.
Given a resource $r$, we will write $A(r)$ to denote the set of users that are authorized to access resource $r$.
More formally, $A(r) = \set{u \in U : (u,r) \in A}$.
Similarly, for $u \in U$, we will write $A(u)$ to denote the set of resources that $u$ is authorized to access, that is $A(u) = \set{r \in R : (u,r) \in A}$.
We extend this notation to subsets of $R$ and $U$ in the natural way: for $R' \subseteq R$ and $U' \subseteq U$, 
\[
 A(R') \stackrel{\rm def}{=} \bigcup_{r \in R'} A(r) \quad\text{and}\quad A(U') \stackrel{\rm def}{=} \bigcup_{u \in U'} A(u). 
\]

We introduce two fundamental concepts, which will be used to formulate the {\sc Authorization Policy Existence Problem}.
\begin{itemize}
 \item a \emph{base authorization relation} $\baserel \subseteq U \times R$ such that $\baserel(r)\neq \emptyset$ for each $r\in R$;
%  \item a \emph{target authorization relation} $\targetrel$; and
 \item a set of \emph{authorization constraints} $C$.
\end{itemize}
Informally, $\baserel$ specifies restrictions on all valid authorization relations, while $C$ specifies additional restrictions that any valid authorization relation must satisfy.

An authorization constraint may be defined by enumerating the set of all authorization relations that satisfy the constraint.
% Thus, we say $\targetel$ satisfies a constraint $c$ if $A \in c$.
Of course, an extensional definition of this nature is utterly impractical, and all useful constraints will be defined in an intensional manner.
A simple example would be a constraint requiring no user is assigned to both resources $r$ and $r'$.
In other words, an authorization relation $A$ satisfies this constraint provided that $\set{(u,r),(u,r')} \not\subseteq A$ for all $u \in U$.
We discuss constraints in more detail in Section~\ref{sec:constraints}.
% Thus, a \emph{valid} target authorization relation must be \emph{authorized}, \emph{eligible} and \emph{complete}.

More formally, we have the following definitions.
Given a base authorization relation $\baserel$ and a set of constraints $C$, we say an authorization relation $A \subseteq U \times R$ is 
  \begin{itemize}
   \item \emph{authorized} with respect to $\baserel$ if $A \subseteq \baserel$;
   \item \emph{complete} if $A(r)\neq \emptyset$ for every $r \in R$;
   \item \emph{eligible} with respect to $C$ if $A$ satisfies $c$ for all $c \in C$; and
   \item \emph{valid} with respect to $\baserel$ and $C$ if $A$ is authorized, complete and eligible.
  \end{itemize}

We introduce the term {\sc Authorization Policy Existence Problem} (\APEP) as a generic term for questions related to finding a valid authorization relation, given $\baserel$ and $C$.
Then \APEP comes in two ``flavors'':
\begin{description}
 \item[Decision (D-\APEP):] Does there exist a valid authorization relation? If so, find a valid authorization relation.
 %\item[Search (S-\APEP):] 
 \item[Optimization (O-\APEP):] Find a ``best'' valid authorization relation if one exists (where the objective function has to be specified).
\end{description}
We assume that determining whether an authorization relation satisfies a constraint takes polynomial time.
(This is a reasonable assumption for all constraints of relevance to access control.)
Let $n$ denote $\card{U}$, $k$ denote $\card{R}$ and $m$ denote $\card{C}$.
Then a brute force approach to solving D-\APEP (by simply examining every possible authorization relation) takes time $O^*(2^{nk})$. 
%In what follows, we may assume that $k\le n$ as otherwise there is no complete authorization relations. -> wrong according to me (Pierre)
%, where the notation $O^*$ ignores polynomial terms (in $n$, $k$ and $m$).

A few special cases of \APEP are worth mentioning.
For simplicity we will write $\targetrel$ to denote one of the authorization relations that are solutions to \APEP.
\begin{enumerate}
 \item $\targetrel$ is required to be a \emph{function} $\targetrel : R \rightarrow U$.  
       In this case, we allocate a unique user to each resource.
       Computing a plan allocating one user to each step in a workflow instance, subject to an authorization policy defined by $\baserel$ and some constraints $C$, is an example of this type of scenario.
       In this case, D-\APEP corresponds to the {\sc Workflow Satisfiability Problem} (WSP)~\cite{WaLi10}.
       Moreover, the {\sc Cardinality-Constrained Minimum User Problem}~\cite{RoSuMaVaAt15}, whose solution is a plan using the minimum number of users, is an instance of O-\APEP.
 \item $\baserel = U \times R$.  In this case, $\baserel$ itself imposes no restrictions on $\targetrel$.  
       All restrictions on $\targetrel$ are defined by the constraints $C$.
       Defining an authorization policy in the presence of static separation-of-duty constraints is an example of this type of scenario.
       Li, Tripunitara and Bizri have studied problems of this nature~\cite{LiTrBi07}.
 \item A constraint that requires each resource is assigned to at least $t$ users enables us to model problems related to resiliency in access control~\cite{LiWaTr09} and workflow systems~\cite{WaLi10}.
 \item If we seek to maximize the cardinality of $\targetrel$, then, informally, a solution to O-\APEP provides a ``resilient'' authorization policy.
       While this is different from existing notions of resiliency~\cite{LiWaTr09,WaLi10}, it would seem to be an interesting way of approaching the problem of making an authorization policy resilient to the unavailability of users.
\end{enumerate}

\section{Constraints}\label{sec:constraints}

We now describe constraints in more detail.
We generalize the approach used in prior work on constraints for workflow systems~\cite{CrGuYe13,WaLi10}.

\subsection{Binding-of-duty and separation-of-duty\\ constraints}\label{sec:bod-and-sod-constraints}
% 
% Let $\sim$ be either $=$ or $\neq$, and $Q$ be either $\forall$ or $\exists$. 
% Given $r, r' \in R$, we define the constraint $(r, r', \sim, Q)$, which is satisfied by a relation $A \subseteq U \times R$ if and only if the following proposition is true:
% $$Q u \in A(r) \cup A(r'), \left[(u, r) \in A \right] \sim \left[ (u, r') \in A \right]$$
% Here, $\sim$ (\emph{i.e.} either $=$ or $\neq$) has to be interpreted as a relation between two logical propositions.
% If $Q = \forall$ (resp. $Q = \exists$), then we call such constraints ``universal'' (resp. existential).
% 
\sloppy
\emph{Binding-of-duty} and \emph{separation-of-duty} constraints have received considerable attention in the access control literature, and such constraints may be \emph{static} or \emph{dynamic}.
Informally, static constraints specify restrictions on policy relations, whereas dynamic constraints specify constraints on particular sequences of events within the context of an access control system.
A (static) separation-of-duty constraint, for example, in the context of role-based access control system, might require that no user is authorized for both roles $r$ and $r'$~\cite{SaCoFeYo96}.
In contrast, a (dynamic) separation-of-duty constraint, in the context of a workflow system, might simply require that two steps $s$ and $s'$ are performed by different users in each instance of the workflow~\cite{Cr05,WaLi10}.
(This constraint, however, does not prevent the same user being authorized for both those steps.)
Within the framework of \APEP, we seek to define a more general (and uniform) syntax and semantics for constraints.

We express constraints in terms of the following logical (binary) operators defined via their respective truth tables:
\[
 \begin{array}{c|c|c|c|c|c}
  p & q & p \leftrightarrow q & p \rightarrow q & p \leftarrow q & p \updownarrow q \\
 \hline
  0 & 0 & 1 & 1 & 1 & 0 \\
  0 & 1 & 0 & 1 & 0 & 1 \\
  1 & 0 & 0 & 0 & 1 & 1 \\
  1 & 1 & 1 & 1 & 1 & 0 \\
 \end{array}
\]
Let $r$ and $r'$ be resources in $R$; let $\circ$ denote one of the logical operators in the set $\set{\leftrightarrow,\leftarrow,\rightarrow,\updownarrow}$; and let $Q$ be one of the first-order quantifiers $\exists$ or $\forall$.
Then $(r,r',\circ,Q)$ is a constraint: a constraint of the form $(r,r',\circ,\forall)$ is said to be \emph{universal}, while a constraint of the form $(r,r',\circ,\exists)$ is said to be \emph{existential}.
A complete relation $A$ 
 \begin{itemize}
  \item satisfies $(r,r',\circ,\exists)$ if there exists $u \in A(r) \cup A(r')$ such that the propositional formula $(u \in A(r)) \circ (u \in A(r'))$ evaluates to true; and
  \item satisfies $(r,r',\circ,\forall)$ if for all $u \in A(r) \cup A(r')$, the propositional formula $(u \in A(r)) \circ (u \in A(r'))$ evaluates to true.
 \end{itemize}
Note that for any complete relation $A$ and any $r \in R$, $A(r) \ne \emptyset$, so $A(r) \cup A(r') \ne \emptyset$.
Thus constraints cannot be vacuously satisfied by a valid relation.

Informally speaking, universal constraints are ``stronger'' than existential constraints: (for any complete relation) the satisfaction of $(r,r',\sim,\forall)$ implies the satisfaction of $(r,r',\sim,\exists)$, but the converse does not hold.
%Again, in an informal sense, a universal constraint corresponds to a static constraint by fixing certain aspects of any relation $A$ that satisfies the constraint.
%Conversely, an existential constraint may admit many different relations that satisfy it.

We now expand these generic definitions for the four constraints defined by $(r,r',\circ,Q)$, where $\circ$ is either $\leftrightarrow$ or $\updownarrow$:
\begin{enumerate}
 \item $(r,r',\leftrightarrow,\exists)$ is satisfied if there exists $u \in U$ such that $u \in A(r)$ and $u \in A(r')$; that is, $A(r) \cap A(r') \ne \emptyset$.
 \item $(r,r',\updownarrow,\exists)$ is satisfied if there exists $u \in U$ such that either \begin{inparaenum}[(i)]\item $u \in A(r)$ and $u \not\in A(r')$ or \item $u \not\in A(r)$ and $u \in A(r)$\end{inparaenum};
       that is, $A(r) \ne A(r')$.
 \item $(r,r',\leftrightarrow,\forall)$ is satisfied if for all $u \in A(r) \cup A(r')$, $u \in A(r)$ if and only if $u \in A(r')$;
       that is, $A(r) = A(r')$.
 \item $(r,r',\updownarrow,\forall)$ is satisfied if for all $u \in A(r) \cup A(r')$, either \begin{inparaenum}[(i)]\item $u \in A(r)$ and $u \not\in A(r')$ or \item $u \not\in A(r)$ and $u \in A(r')$\end{inparaenum}; 
       that is, $A(r) \cap A(r') = \emptyset$.
\end{enumerate}
Thus, constraints of the form $(r,r',\updownarrow,Q)$ correspond closely to the idea of separation-of-duty.
Indeed, the satisfaction criterion for $(r,r',\updownarrow,\forall)$ is identical to that for a simple static separation-of-duty constraint.
Similarly, constraints of the form $(r,r',\leftrightarrow,Q)$ correspond to the idea of binding-of-duty.

Now consider a constraint of the form $(r,r',\rightarrow,\forall)$.
Such a constraint is satisfied if for all $u \in A(r) \cup A(r')$, $(u \in A(r)) \rightarrow (u \in A(r'))$.
In other words, $A(r) \subseteq A(r')$.
Thus a global constraint of this form could be used to specify a role hierarchy (in which role $r'$ is senior to $r$).
Conversely, a constraint of the form $(r,r',\leftarrow,\forall)$ could be used to specify a role hierarchy in which $r'$ is junior to to $r$.%
  \footnote{Since $A(r)$ and $A(r')$ are non-empty, the constraints $(r,r',\leftarrow,\exists)$ and $(r,r',\rightarrow,\exists)$ are both equivalent to \mbox{$(r,r',\leftrightarrow,\exists)$}.}
Thus we can use constraints to insist that a hierarchy is strict: that is, there exists at least one user that is assigned to $r'$ but not $r$.
Specifically, a relation $A$ simultaneously satisfies constraints $(r,r',\rightarrow,\forall)$ and $(r,r',\updownarrow,\exists)$ only if $A(r) \subset A(r')$.

\subsection{Cardinality constraints}

We may also define \emph{cardinality constraints}, which come in two flavors. In the following, $\varineq$ is one of $=$, $<$, $>$, $\leqslant$ or $\geqslant$ and $t$ is an integer greater than $0$.  
 \begin{itemize}
  \item A \emph{global} (cardinality) constraint has the form $(\varineq,t)$.
	The constraint $(\varineq,t)$ is satisfied by relation $A$ if for all $r \in R$, $\card{A(r)} \varineq t$.
  \item A \emph{local} (cardinality) constraint has the form $(R',\varineq,t)$, where $R' \subseteq R$.
	The constraint $(R',\varineq,t)$ is satisfied by relation $A$ if $\card{A(R')} \varineq t$.
 % \item ??? Do we want constraints defined by $R'$, $\varineq$ and $t$ whose satisfaction criterion is  $\sum_{r \in R'} \card{A(r)} \varineq t$???
 \end{itemize}
 \sloppy Then, for example, the global constraint $(=,1)$ requires a valid relation $A$ to be a function (since the number of users assigned to each resource is precisely $1$), while the local constraint $(\set{r},\leqslant,t)$ is a cardinality constraint in the RBAC96 sense~\cite{SaCoFeYo96} (if resource $r$ is interpreted as a role).
 Finally, the $p$-out-of-$q$ static separation-of-duty constraint \mbox{${\sf ssod}(\set{r_1,\dots,r_q},p)$}, introduced by Li \emph{et al.}~\cite{LiTrBi07}, may be represented by the cardinality constraint \mbox{$(\set{r_1,\dots,r_q},\geqslant,p)$}.

 \begin{remark}\label{rem1}
  If we define a global constraint $(=,1)$, then the universal constraint $(r,r',\circ,\forall)$ is equivalent to the existential constraint $(r,r',\circ,\exists)$ (in the sense that the former is satisfied if and only if the latter is).
 \end{remark}
 
 \subsection{User-independent constraints}
 
 Research on the {\sc Workflow Satisfiability Problem} has shown that the notion of \emph{user-independent} (UI) constraints is very important.
 First, the class of UI constraints includes a very wide range of constraints and almost all constraints that are of relevance to access control.
 Second, WSP is fixed-parameter tractable (FPT) if we restrict attention to UI constraints~\cite{CoCrGaGuJo14}.
 (WSP is not FPT if we allow arbitrarily complex constraints~\cite{WaLi10}.)
 Informally, a constraint is UI in the context of workflow satisfiability if its satisfaction only depends on the relationships that exist between users assigned to steps in a workflow (and not on the specific identities of users)~\cite{CoCrGaGuJo14}.
 We now extend the definition of user-independent used in the context of workflow satisfiability.
 
 \sloppy
 Let $A$ be an authorization relation and $\sigma : U \rightarrow U$ a permutation of the user set (that is, $\sigma$ is a bijection).
 Then, given an authorization relation $A \subseteq U \times R$, we write $\sigma(A) \subseteq U \times R$ to denote the relation \mbox{$\set{(\sigma(u),r) : (u,r) \in A}$}.
 A constraint $c$ is \emph{user-independent} if for every authorization relation $A$ that satisfies $c$ and every permutation $\sigma : U \rightarrow U$, $\sigma(A)$ satisfies~$c$.
 
 \fussy
 Consider, for example, a constraint of the form $(r,r',\leftrightarrow,\exists)$ and suppose that $A \subseteq U \times R$ satisfies the constraint.
 Then, by definition, there exists a user $u$ such that $(u,r),(u,r') \in A$.
 Then, for any permutation $\sigma$, $(\sigma(u),r),(\sigma(u),r') \in \sigma(A)$, so $\sigma(A)$ also satisfies the constraint.
 Similar arguments may be used to show that constraints of the form \mbox{$(r,r',\leftrightarrow,\forall)$}, \mbox{$(r,r',\updownarrow,\exists)$} and $(r,r',\updownarrow,\forall)$ are all UI.
 Equally, it is clear that global and local constraints, whose satisfaction is defined in terms of the cardinality of sets of the form $A(r)$, are UI, since a permutation (being a bijection) will preserve the cardinality of such sets.
 In other words, all constraints we consider in this paper are UI.

\subsection{Bounded UI constraints} \label{subsec:bcc}

\renewcommand{\ker}{\mathsf{core}}

We now define an important class of UI constraints that will be useful for establishing positive results in the remainder of the paper.
Given a base relation $\baserel$ and a constraint $c$, let $A$ be valid with respect to $\baserel$ and $c$.
We say $A$ \emph{requires $v$} if $\set{(u,r) \in A : u \ne v}$ is not valid.
(Since $A$ is valid, this means that $\set{(u,r) \in A : u \ne v}$ is either incomplete or does not satisfy $c$.)
Then we define \[ \ker(A : \baserel,c) \stackrel{\rm def}{=} \set{u \in U : \text{$A$ requires $u$}} \] to be the \emph{core of $A$ with respect to $\baserel$ and $c$}.
% and write $\ker(A)$ to denote the cardinality of the kernel of $A$.

% Thus the core of $A$ represents those users that are required for $A$ to be valid.
Consider for instance a constraint $c$ of the form $(r,r',\updownarrow,\forall)$. If there exists an authorization relation $A$ satisfying $c$, then there is one whose core contains at least two users: indeed, remove iteratively from $A$ any couple $(u, r) \in A$ such that $A$ does not require $u$. When this is no longer possible, the obtained relation has to allocate two distinct users to $r$ and $r'$, both belonging to the core.
% any core must contain at least two users, since $r$ and $r'$ must be assigned to different users in a valid relation.
The core could contain as many as $k$ users, since the relation may allocate each resource to a different user and removing any user would compromise the completeness of the relation.
However, the core cannot contain more than $k$ users, since in any set of at least $k+1$ users, at least two users must be allocated to the same resource and one of them could be removed without compromising the completeness or the eligibility of the relation.
Conversely, for a constraint of the form $(r,r',\leftrightarrow,\forall)$ the core could contain a single user but no core contains more than $k-1$ users, since $r$ and $r'$ must be assigned to the same user and any additional users may be removed without compromising completeness or eligibility.
% 
% \begin{proposition}
%  Let $\I = (\baserel,C)$ be an instance of \DAPEP such that all constraints in $C$ are UI.
%  Then for any solution $A$ of $\I$ and any constraint $c \in C$, there exists a relation $A'$ such that 
%  \[
%   |\ker(A' : U \times R,c)| \geqslant |\ker(A : \baserel,C)|.
%  \]
% \end{proposition}

\begin{proposition}\label{pro:kernel-maximized-by-full-relation}
Let $\I = (\baserel, C)$ be a satisfiable instance of \DAPEP with a UI constraint $c \in C$. 
If $A$ is a valid solution with respect to $\baserel$ and $c$ then $$|\ker(A : U \times R, c)| \ge |\ker(A : \baserel, c)|$$
\end{proposition}

\begin{proof}
\sloppy
We prove the more general statement that $\ker(A : \baserel, c) \subseteq  \ker(A : \Omega, c)$ for any $\Omega \supseteq \baserel$. 
Suppose $A$ requires $u$ and let $A \setminus u$ denote $\set{(v, r) \in A : u \neq v}$.
Then $A \setminus u$ is either incomplete or violates $c$ (since $A$ is authorized, so is $A \setminus u$). 
If $A \setminus u$ is incomplete, then it is also incomplete for the instance $(\Omega, c)$. 
The same argument holds if $A \setminus u$ violates $c$, which concludes the proof.
\end{proof}

\begin{definition}
\sloppy
We say a UI constraint $c$ is $f(k,n)$-\emph{bounded} if $\card{\ker(A : U \times R,c)} \leqslant f(k,n)$ for all $A$ valid with respect to $U \times R$ and $c$.
\end{definition}
% Finally, we write 
% \[
%  \ker(\baserel,c) = \max\set{\ker(A) : \text{$A$ is valid w.r.t. $\baserel$ and $c$}}
% \]
% By definition, $\ker(A,\baserel,C)$ imposes an upper limit on the number of users that we need to consider when constructing candidate solutions to an instance $(\baserel,C)$ of \DAPEP.
The definition of $f(k,n)$-bounded constraints and Proposition~\ref{pro:kernel-maximized-by-full-relation} impose an upper bound on the number of users we need to consider when constructing candidate solutions to an instance $(\baserel,C)$ of \DAPEP.
% If all constraints are $f(k)$-bounded then the number of candidate solutions that we need to consider depends on the smaller parameter.

We have proved several results which establish $f(k,n)$ for a number of constraint types.
These results are summarized in Table~\ref{tbl:bcc-constraints}; proofs can be found in the appendix.
Note that in all cases, $f(k,n)$ is independent of $n$.
This is important as we show in Section~\ref{sec:apep-complexity} that \DAPEP is FPT when all constraints are $f(k)$-bounded for some function $f$. 

\begin{table}[h]\centering\setlength{\extrarowheight}{2pt}
 \begin{tabular}{|r|r|}
 \hline
  \bf Constraint Type & \bf Largest Core \\
 \hline
  $(r,r',\updownarrow,\forall)$, $(r,r',\updownarrow,\exists)$ & $k$ \\
%   $(r,r',\updownarrow,\exists)$ & $k$ \\
  $(r,r',\leftrightarrow,\forall)$, $(r,r',\rightarrow,\forall)$, $(r,r',\leftrightarrow,\exists)$ & $k-1$ \\
%   $(r,r',\rightarrow,\forall)$ & $k-1$ \\
%   $(r,r',\leftrightarrow,\exists)$ & $k-1$ \\
 \hline
  $(R',\leq,t)$ & $k$ \\
  $(R',=,t)$, $(R',\geq,t)$ & $2\max\set{k,t}$ \\
%   $(R',\geq,t)$ & $2\max\set{k,t}$ \\
 \hline
 \end{tabular}
 \caption{Upper bounds on the size of the core}\label{tbl:bcc-constraints}
\end{table}

 \subsection{Notation}

 In the remainder of this paper, we consider versions of \APEP in which we restrict our attention to particular types of constraints.
 We use the following abbreviations for families of constraints:%
 \begin{inparaenum}[(i)]
  \item $\bod$ to denote the family of all existential and universal (binding-of-duty) constraints having the form $(r,r',\leftrightarrow,\exists)$ or $(r,r',\leftrightarrow,\forall)$;
  \item $\sod$ to  denote the family of all existential and universal (separation-of-duty) constraints having the form $(r,r', \updownarrow,\exists)$ or $(r,r',\updownarrow,\forall)$;
  \item $\bode$ and $\bodu$ to denote, respectively the family of all existential and universal binding-of-duty constraints;
  \item $\sode$ and $\sodu$ to denote, respectively the family of all existential and universal separation-of-duty constraints;
  \item $f(k)$-bounded to denote $f(k)$-bounded-certificate constraints;
  \item $\gbl$ to denote the family of all global cardinality constraints;  
  \item $\lcl$ to denote the family of all local cardinality constraints.
  %\item $\fnc$ to denote the constraint $(=,1)$, which requires the solution $\targetrel$ to be a function $\targetrel : R \rightarrow U$.
 \end{inparaenum}
 Finally, we write, for example, $\APEParg{ \bod, \lcl }$ to restrict the set of instances of \APEP in which the set of constraints $C$ contains only binding-of-duty and local cardinality constraints.

\section{Complexity of APEP}\label{sec:apep-complexity}

Before exploring the fine-grained complexity of $\DAPEP$ with respect to the different types of constraints, we first state general properties about some special cases.

Firstly, note that adding the function constraint $\fnc$ to any $\DAPEP$ instance having only SoD or BoD constraints forces any solution $\targetrel$ of $\DAPEP$ to be a function. In this case, $\DAPEP$ becomes equivalent to WSP. More formally, we say that a parameterized problem $A$ is \emph{parameter-reducible} to a parameterized problem $B$ (and we write $A \le_{\it fpt} B$) if there is a polynomial algorithm which transforms an instance $(I, k)$ of $A$ into an instance $(I', k')$ of $B$ such that $(I, k)$ is positive for $A$ iff $(I', k')$ is positive for $B$, and $k' \le f(k)$ for some computable function $f$. Clearly, if $A$ is parameter-reducible to $B$ and $B$ is FPT, then $A$ is FPT. We say that $A$ is \emph{parameter-equivalent} to $B$ (and we write $A =_{\it fpt} B$) if $A \le_{\it fpt} B$ and $B \le_{\it fpt} A$. The proof of the following result is straightforward, by the definition of the function constraint $\fnc$.
\begin{theorem}\label{thm:APEPequivWSP}
If $\DAPEP$ and WSP are parameterized by the number of resources and steps, respectively, we have the following:
\begin{itemize}
	\item $\DAPEParg{ \sod, \fnc }$ $=_{\it fpt}$ WSP$(\neq)$;
	\item $\DAPEParg{ \bod, \fnc }$ $=_{\it fpt}$ WSP$(=)$; and
	\item $\DAPEParg{ \bod, \sod, \fnc }$ $=_{\it fpt}$ WSP$(=, \neq)$.
\end{itemize}
\end{theorem}

Moreover, the following result asserts that adding $\bodu$ constraints to any instance of $\DAPEP$ does not change its (parameterized) complexity.

\begin{theorem}\label{bodu-th}
\sloppy
Given any instance $(U, R, \baserel, C)$ of $\DAPEParg{\bod, \sod}$, one can obtain in polynomial time an instance $(U, R', \baserel', C')$ of $\DAPEP$ such that:
\begin{inparaenum}[(i)]
	\item $C'$ does not contain any $\bodu$ constraint,
	\item $|\baserel'| \le |\baserel|$,
	\item $|R'| \le |R|$, and
	\item $|C'| \le |C|$.
\end{inparaenum}
\end{theorem}
\begin{proof} Let $C^*$ be the set of $\bodu$ constraints from $C$. 
The idea is to consider $\bodu$ constraints as an equivalence relation: for $r, r' \in R$, $r \sim_b r'$ if and only if $\left(r,r',\leftrightarrow,\forall \right) \in C^*$. 
Now, let $R' = \{R_1, \dots, R_q\}$ be the equivalence classes of $\sim_b$. 
Obviously, $|R'| \le |R|$. For all $i \in [q]$ and all $r \in R_i$, set $\baserel'(r) = \bigcap_{r' \in R_i} \baserel(r')$. 
Once again, it holds that $|\baserel'| \le |\baserel|$.
Finally, for every constraint $c=(r, r', \sim, Q) \in C\setminus C^*$ with $\sim \in \{=, \neq\}$ and $Q \in \{\forall, \exists\}$, we distinguish two cases:
\begin{itemize}
	\item if $r \in R_i$ and $r' \in R_j$ with $i \neq j$, then add the constraint $(R_i, R_j, \sim, Q)$ (notice that in this case $c$ is not a $\bodu$ constraint);
	\item if $r, r' \in R_i$ for some $i \in [k]$, then if $c$ is a $\sodu$ or $\sode$ constraint, obviously the instance is unsatisfiable (and we can output a trivially negative instance of $\DAPEP$).
\end{itemize}
Clearly $|C'| \le |C|$, and $C'$ does not contain any $\bodu$ constraint. Finally, one can check that the output instance is satisfiable if and only if the input instance is satisfiable.
\end{proof}

\begin{remark}
A corollary of Theorem \ref{bodu-th} is that we may exclude $\bodu$ constraints from consideration when establishing FPT results.
\end{remark}

In the remainder of this section, we establish that \DAPEP with bounded constraints is FPT. 
We also propose extra algorithms for some mixed policies composed of $\bod$ and $\sod$ constraints in order to improve the execution time of the main algorithm for these subcases. 
Figure~\ref{summary} summarizes our results for \DAPEP with $\bod$ and $\sod$ constraints. 

\begin{figure}[!bth]
\centering

\begin{tikzpicture}
\node[draw,align=center,minimum width=3.2cm] (besuse) at (3.0,7.0) {$\bode,\sodu,\sode$\\ \textcolor[rgb]{0,0,1}{FPT: $2^{2^k k^2}$}};
\node[draw,align=center,minimum width=2.3cm] (be) at (0.5,3.0) {$\bode$\\ \textcolor[rgb]{0,0,1}{P}};
\node[draw,align=center,minimum width=2.3cm] (su) at (3.0,3.0) {$\sodu$\\ \textcolor[rgb]{0,0,1}{FPT: $2^{k}$}};
\node[draw,align=center,minimum width=2.3cm] (se) at (5.5,3.0) {$\sode$\\ \textcolor[rgb]{0,0,1}{FPT: $2^{k}$}};
\node[draw,align=center,minimum width=2.3cm] (besu) at (0.5,5.0) {$\bode,\sodu$\\ \textcolor[rgb]{0,0,1}{FPT: $2^{k^2\log k^2}$}};
\node[draw,align=center,minimum width=2.3cm] (bese) at (3.0,5.0) {$\bode,\sode$\\ \textcolor[rgb]{0,0,1}{FPT: $2^{2^k k^2}$}};
\node[draw,align=center,minimum width=2.3cm] (suse) at (5.5,5.0) {$\sodu,\sode$\\ \textcolor[rgb]{0,0,1}{FPT: $2^{2^k k^2}$}};

\draw (besu) -- (besuse);
\draw (suse) -- (besuse);
\draw (bese) -- (besuse);
\draw (be) -- (besu);
\draw (su) -- (besu);
\draw (be) -- (bese);
\draw (se) -- (bese);
\draw (su) -- (suse);
\draw (se) -- (suse);

\end{tikzpicture}

\caption{Complexity of specific cases of D-APEP (polynomial factors are ignored)}
\label{summary}
\end{figure}
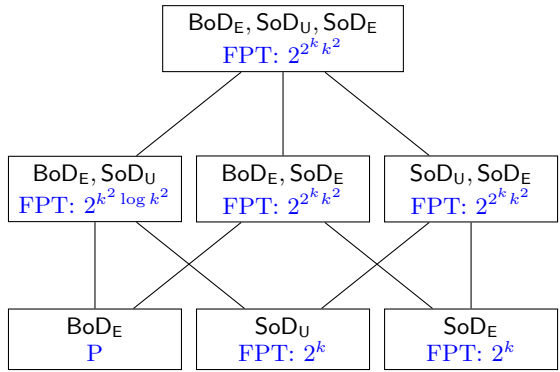

%\begin{remark}
%In fact, \DAPEP can be considered as a ``symmetric'' problem. Users and resources play a similar role. If \DAPEP is FPT parameterized by $k$ with specific UI constraints, then it is also FPT parameterized by $n$ with the transposed RI %constraints. For instance, as \DAPEP is FPT parameterized by $k$ with $\sodu$ constraints $(r,r',\updownarrow,\forall)$, \DAPEP is also FPT parameterized by $n$ with constraints like $(u,u',\updownarrow,\forall)$. 
%\end{remark}

\subsection{Instances with constraints of a single type}

As a direct application of Theorem \ref{bodu-th}, we are able to show that  $\DAPEParg{ \bodu }$ is polynomial-time solvable: indeed, after using the reduction of the previous result, it is clear that the instance is satisfiable if and only if $\baserel'$ is complete.

\begin{theorem}
 $\DAPEParg{ \bodu }$ is solvable in polynomial time.
\label{boduth}
\end{theorem}

We now consider the complexity of $\DAPEParg{\sodu}$ and \mbox{$\DAPEParg{\bode}$}.
% instances containing only $\sodu$ constraints and instances containing only $\bode$ constraints.

\begin{theorem}\label{sodufnc}
 $\DAPEParg{ \sodu }$ $=_{\it fpt}$  $\DAPEParg{ \sodu, \fnc}$.
\end{theorem}

\begin{proof}
It is sufficient to prove that there always exists a feasible solution of $\DAPEParg{ \sodu }$ which satisfies the function constraint $\fnc$. %By Theorem~\ref{thm:APEPequivWSP}.

Let us suppose $(U, R, \baserel,C)$ is satisfiable, so there exists a valid solution $\targetrel$ for this instance. 
We define another relation $A'$ which is a function. 
For any $r \in R$, we have $\card{\targetrel(r)} \geq 1$, thus we can pick an arbitrary user $u_r$ in $\targetrel(r)$ and set $A'(r) = \set{u_r}$. 
One can observe that $A' \subseteq \targetrel \subseteq \baserel$, and thus $A'$ is indeed a function. For a constraint $\left(r,r',\updownarrow,\forall\right) \in C$, we have $\targetrel(r) \cap \targetrel(r') = \emptyset$, but since $A'(r) \subseteq \targetrel(r)$ and $A(r') \subseteq \targetrel(r')$, we have $A'(r) \cap A'(r') = \emptyset$ as well. Thus, $A'$ is valid for $(U, R, \baserel,C')$ and is clearly a function.
\end{proof}

Since, by Theorem~\ref{thm:APEPequivWSP}, $\DAPEParg{ \sodu, \fnc }$ $=_{\it fpt}$ WSP$(\neq)$, and since WSP$(\neq)$ is NP-hard and FPT parameterized by the number of steps~\cite{CrGuYe13}, we obtain the following corollary:

\begin{corollary}
 $\DAPEParg{ \sodu }$ is NP-hard and FPT parameterized by $k$.
\end{corollary}

In fact, it follows from \cite{CrGuYe13} that $\DAPEParg{ \sodu }$ can be solved in time $O^*(2^k)$. We now consider $\bode$ constraints.

\begin{theorem}\label{thm:bode}
 $\DAPEParg{ \bode }$ is polynomial-time solvable.
\end{theorem}

\begin{proof}
We show that an instance $(U, R, \baserel, C)$ is satisfiable iff $\baserel$ is valid.
Obviously, if $\baserel$ is valid, the $\DAPEParg{ \bode }$ instance is satisfiable. Conversely observe first that $\baserel$ is obviously authorized and complete. Then, if $\DAPEParg{ \bode }$ is satisfiable, there exists $\targetrel \subseteq \baserel$, which is valid. However, since $\targetrel(r) \subseteq \baserel(r)$ for any $r \in R$, any constraint $(r, r', \leftrightarrow ,\exists)$ satisfied by $\targetrel$ is also satisfied by $\baserel$. In other words, $\baserel$ is eligible.  
\end{proof}

\subsection{Complexity of $\DAPEParg{f(k)-bounded}$}\label{sec:complexity-fk-bounded}

\sloppy
Let $f$ be an arbitrary function in $k$ and let $\I=\left(U,R,\baserel,C\right)$ be a $\DAPEP$ instance composed of $f(k)$-bounded constraints. Without loss of generality, we assume that $f(k) \ge k$ (observe that all constraints considered in this paper are never better than $(k-1)$-bounded). In this section, we introduce a method to decide whether $\I$ is satisfiable or not in time $\mathcal{O}^*\big(2^{2^k f(k) k}\big)$.

Given a set of resources $T \subseteq R$, we define % the family $U_T$ which contains all users $u$ such that $\baserel(u) = T$:
\[
U_T = \left\lbrace u \in U:\ \baserel(u) = T \right\rbrace 
\]
We call $U_T$ the (user) \emph{family} associated with $T$.
Note that for $T \ne T'$, we have $U_{T} \cap U_{T'} = \emptyset$.
Moreover, $U = \bigcup_{T \subseteq R} U_T$.
Thus $\set{U_T}_{T \subseteq R}$ is a partition of $U$ containing at most $2^k$ sets.
% So, all sets $U_T$ are disjoint and cover the entire set $U$: $U = \bigcup_{T \subseteq R} U_T$. 
The intuition behind this definition is that when considering UI constraints (in particular, $f(k)$-bounded constraints), all users in a set $U_T$ play the same role.
%in the sense that given any bijection $\phi : U_T \rightarrow U_R$, and any valid solution $A$, the solution $A'$ that assigns, for any  if a solution assigns user $u_1 \in U_T$ to a resource $r$, then  are interchangeable in terms of constraint satisfaction.
% , and $\bcc$ is a subclass of UI constraints. 
The idea of the algorithm is thus to eliminate users in ``large'' families to upper-bound the number of users by a function of $k$. 
To eliminate users, we apply the following reduction rule:

\begin{quote}
if there exists $T \subseteq R$ such that $|U_T| > f(k)$, then remove an arbitrary user $u^*$ from $U_T$.
\end{quote}
Successive applications of this rule will result in an instance in which the number of users is at most $f(k) 2^k$, a function of $k$ only.

Consider, for example, an instance comprising the base authorization relation $\baserel$ shown in Figure~\ref{algobcc} ($k=3$ and $n=8$) and a single constraint $(r_1,r_2,\updownarrow,\exists)$.
The constraint $(r_1,r_2,\updownarrow,\exists)$ is 3-bounded. 
There are three families of users, of which $U_{\set{r_1,r_2}}$ has cardinality greater than $3$. 
Applying the reduction rule, we may remove users $u_3$ and $u_4$ from $U_{\set{r_1,r_2}}$.
% and then solve the reduced instance with the \textsc{brute force} method.

\begin{figure}[h]
\[
 \begin{array}{r|ccc}
   & r_1 & r_2 & r_3 \\
  \hline
  u_1, u_2 & 1 & & 1 \\
  \cancel{u_3}, \cancel{u_4},u_5,u_6,u_7 & 1 & 1 & \\
  u_8 & 1 & 1 & 1 \\
 \end{array}
\]
\caption{Use of the reduction rule}
\label{algobcc}
\end{figure}
% 
% Then, by applying a \textsc{brute force} method, we obtain an FPT algorithm. 

% Let $\I'$ be the instance obtained after application of the previous reduction rule. We now prove that the reduction rule is safe:

\begin{lemma}
$\I$ is satisfiable iff $\I'$ is satisfiable, where $\I'$ is the instance obtained by applying the reduction rule to $\I$.
\end{lemma}

\begin{proof}
Obviously, if $\I'$ is satisfiable, then so is $\I$.

Assume $\I$ satisfiable and $\I'$ unsatisfiable, and let $A$ be a solution for $\I$. 
Then there exists $T\subseteq R$ such that $|U_T| \ge f(k)+1 \ge k+1$, which means that $A$ violates some constraint $c \in C$ (\emph{i.e.} unsatisfiability of $\I'$ does not come from incompleteness). 
Since $c$ is $f(k)$-bounded, we have $|\ker(A : \I)| \le f(k)$ and thus $|\ker(A : \I) \cap U_T| \le f(k)$. But, since $|U_T| \ge f(k) +1$, and since $c$ is user independent, we may assume, without loss of generality, that there is a user $u^* \in U_T$ such that $u^* \notin \ker(A : \I)$. 
However, $u^*$ is a user whose removal makes the instance unsatisfiable, a contradiction.
%Let $\targetrel$ be a valid relation for $\mathcal{I}$.  
%Set $C$ is composed of $m$ $\bcc$ constraints. 
%By definition of $\bcc$ constraints, we can remove up to $f(k)m$ users from $U$ without making the restriction of $\targetrel$ ineligible or incomplete. So, there are, at most, $f(k)m$ users that cannot be removed in every family $U_T$. Given that every family $U_T$ that was reduced 
%contained at least $f(k)m$ users and that all users in a family are indistinguishable, the restriction of $\targetrel$ to space $U' \times R$ is valid for instance $\mathcal{I}'$.
\end{proof}

\begin{theorem}\label{th-boundedc-fpt} For any computable function $f$ depending only on $k$,
$\DAPEParg{f(k)-bounded}$ is FPT parameterized by $k$.
\label{complexity-boundedc}
\end{theorem}

\begin{proof}
Whenever the reduction rule can be applied, we remove one user from the instance. If the rule cannot be applied, then we have an instance with at most $2^kf(k)$ users. Applying a brute force algorithm (by checking every possible relation for the reduced user set), one can check the satisfiability in time $\mathcal{O}^*\left(2^{2^k f(k) k}\right)$.
\end{proof}

\begin{corollary}
$\DAPEParg{\sod,\bod}$ is FPT parameterized by $k$.
\end{corollary}

\begin{proof}
The result follows from Propositions~\ref{th-boundedc-bodu}, \ref{th-boundedc-sodu}, and Theorem~\ref{th-boundedc-fpt}.
%We proved in Propositions~\ref{th-boundedc-bodu} and \ref{th-boundedc-sodu} that the four types of constraints ($\bode$,$\bodu$,$\sode$,$\sodu$) are $k$-$\bcc$. For each pair of resources $(r,r')$, there are at most four  associated constraints. Hence, $m \leq 2k(k-1)$. Then, by Theorem~\ref{complexity-boundedc}, we can solve the reduced instance in time $\mathcal{O}\left(2^{2^{k+1} k^4}\right)$.
\end{proof}

More generally, as we proved in Section~\ref{subsec:bcc} that cardinality constraints with symbol $\le$ or $<$ are $k$-bounded, such constraints can be added to any  $\DAPEParg{\sod,\bod}$ instance without degrading the execution time. Concerning cardinality constraints with symbols $\ge$, $>$ or $=$, we have the following corollary of Proposition~\ref{prop:cardinalitybcc}.

\begin{corollary}  For any computable function $f$ depending only on $k$,
$\DAPEParg{f(k)-bounded, \gbl, \lcl}$ is FPT parameterized by $k$ plus the maximum cardinality of all cardinality constraints.
\end{corollary}

\subsection{Complexity of $\DAPEParg{ \bode,\sodu }$}

We now prove that a better running time can be obtained when considering only $\bode$ and $\sodu$ constraints.

\begin{theorem}
 $\DAPEParg{ \bode,\sodu}$ can be solved in time $\mathcal{O}\left(2^{k^2 \log k^2}\right)$.
\end{theorem}

\begin{proof}
We reduce to WSP$(=, \neq)$. 
Let $R = \set{r_1,\ldots,r_k}$. 
We build a WSP$(=,\neq)$ instance, denoted by $\left(S',U',A',C' \right)$. 
We set $U' = U$. 
Then, for any $i \in [k]$, let \[ \Gamma(r_i) = \{j \in [k]:\ (r_i, r_j,\leftrightarrow, \exists) \in C\}. \] 
For each resource $r_i \in R$, we introduce a set of steps $S^i$. If $\Gamma(r_i) = \emptyset$, then $S^i = \{s^i\}$. 
Otherwise, $S^i = \{s^i_j: j \in \Gamma(r_i)\}$. 
Then define $S' = \bigcup_{i \in [k]} S^i$. 
Observe that $|S'| \leq k(k-1)$.
We then define the following constraints:
\begin{itemize}
	\item For any $i \in [k]$, if $\Gamma(r_i) \neq \emptyset$, then for all $j \in \Gamma(r_i)$, we add the constraint $(s^i_j, s^j_i, =)$.
	\item For any $i, j \in [k]$, if $(r_i, r_j, \updownarrow, \forall) \in C$, then, for every $s \in S^i$ and every $s' \in S^j$, we add the constraint $(s, s', \neq)$.
\end{itemize}
Finally, we define the authorization policy $A'$. For every $i \in [k]$ and every $u \in U$:
\begin{itemize}\sloppy
	\item If $\Gamma(r_i) = \emptyset$, then $(u, s^i) \in A'$ iff $(u, r_i) \in \baserel$.
	\item If $\Gamma(r_i) \neq \emptyset$, then $\forall j \in \Gamma(r_i)$, $(u, s^i_j) \in A'$ iff \mbox{$(u, r_i), (u, r_j) \in \baserel$}.
\end{itemize}
The construction is illustrated in Figure~\ref{fig:bodesodu} for a small example.
Clearly this construction can be carried out in polynomial time.

\begin{figure}[h]

\begin{subfigure}{0.4\columnwidth}
\centering
\begin{tikzpicture}[circle,auto,node distance=1.6cm]
%Graph APEP
\node[draw] (r1) {$r_1$};
\node[draw,below of=r1] (r2) {$r_2$};
\node[draw,right of=r1] (r3) {$r_3$};
\node[draw,right of=r2] (r4) {$r_4$};

\draw (r1) edge node[swap,inner sep=0pt] {$\leftrightarrow$} (r2);
\draw (r1) edge node[inner sep=0pt] {$\leftrightarrow$} (r3);
\draw (r1) edge node[inner sep=0pt] {$\updownarrow$} (r4);
\draw (r2) edge node[swap,inner sep=0pt] {$\updownarrow$} (r4);
% 
% \draw[very thick] (2.8,2.8) -- (3.3,2.8);
% \node[very thick] at (3.8,2.8) {$\bod$};
% \draw[dashed,very thick] (2.8,3.3) -- (3.3,3.3);
% \node[very thick] at (3.8,3.3) {$\sod$};
\end{tikzpicture}
\caption{$C$}
\end{subfigure}
\hfill
\begin{subfigure}{.55\columnwidth}
\centering
\begin{tikzpicture}[circle,auto,node distance=1.6cm]
%Graph WSP
\node[draw] (s12) {$s_2^1$};
\node[draw,below of=s12] (s21) {$s_1^2$};
\node[draw,right of=s12] (s13) {$s_3^1$};
\node[draw,right of=s13] (s31) {$s_1^3$};
\node[draw,right of=s21] (s4) {$s^4$};

\draw (s12) edge node[swap,inner sep=0pt] {$=$} (s21);
\draw (s4) edge node[inner sep=0pt,swap] {$\ne$} (s13);
\draw (s31) edge node[inner sep=0pt,swap] {$=$} (s13);
\draw (s12) edge node[inner sep=0pt] {$\ne$} (s4);
\draw (s21) edge node[swap,inner sep=0pt] {$\ne$} (s4);
% \node[scale=0.9,circle,draw] (s12) at (4,2.0) {$s_2^1$};
% \node[scale=0.9,circle,draw] (s21) at (4.0,1.0) {$s_1^2$};
% \node[scale=0.9,circle,draw] (s13) at (6.0,2.0) {$s_3^1$};
% \node[scale=0.9,circle,draw] (s31) at (6.0,3.2) {$s_1^3$};
% \node[scale=0.9,circle,draw] (s4) at (6,1.0) {$s_4$};
% 
% \draw[very thick] (s12) edge node[inner sep=0pt] {$=$} (s21);
% \draw[dashed,very thick] (s12) edge node[inner sep=0pt] {$\ne$} (s4);
% \draw[dashed,very thick] (s13) edge node[inner sep=0pt] {$\ne$} (s4);
% \draw[dashed,very thick] (s21) edge node[inner sep=0pt] {$\ne$} (s4);
% \draw[very thick] (s13) edge node[inner sep=0pt] {$=$} (s31);
\end{tikzpicture}
\caption{$C'$}
\end{subfigure}

\vspace*{.5\baselineskip}

\begin{subfigure}[h]{0.475\columnwidth}
\centering
\[
\begin{array}{c|cccc}
\baserel & r_1 & r_2 & r_3 & r_4\\
\hline
u_1 & 1 & 1 & 1 & \\
u_2 & & & 1 & 1\\
u_3 & 1 & & 1 & \\
u_4 & 1 & 1 & & \\
u_5 & & 1 & & 1
\end{array} 
\]
\caption{$\baserel$}
\end{subfigure}
\hfill
\begin{subfigure}{.475\columnwidth}
\[
\begin{array}{c|ccccc}
A' & s_2^1 & s_3^1 & s_1^2 & s_1^3 & s^4\\
\hline
u_1 & 1 & 1 & 1 & 1 & \\
u_2 & & & & & 1\\
u_3 & & 1 & & 1 & \\
u_4 & 1 & & 1 & & \\
u_5 & & & & & 1
\end{array} 
\]
\caption{$A'$}
\end{subfigure}

\caption{Reducing $\DAPEParg{\bode,\sodu}$ to WSP$(=,\neq)$}
\label{fig:bodesodu}
\end{figure}

Let us suppose there is a valid plan $\pi: S' \rightarrow U'$ for $\left(S',U',A',C' \right)$. We set $A_{\pi} = \{(u,r_i) : i \in [k], \pi(s) = u $ for some $s \in S^i \}$. One can observe that $A_{\pi}$ is authorized and complete. Then, for any $i, j \in [k]$ such that $\left(r_i,r_j,\leftrightarrow,\exists\right) \in C$, we must have $\pi(s^i_j) = \pi(s^j_i) = u$ for some $u \in U$, which implies that $u \in A_{\pi}(r_i) \cap A_{\pi}(r_j)$. Then, if $\left(r_i,r_j,\updownarrow,\forall\right) \in C$, we know that $\pi(s) \neq \pi(s')$ for any $s \in S^i$ and any $s' \in S^j$, which implies that $A_{\pi}(r_i) \cap A_{\pi}(r_j) = \emptyset$. This proves that $A_{\pi}$ is valid.

Conversely, we suppose $\left(U, R, \baserel,C\right)$ is satisfiable, and let $\targetrel$ be a valid solution. For any $i \in [k]$, we have the following:
\begin{itemize}
	\item If $\Gamma(r_i) = \emptyset$, then define $\pi(s^i)$ to be an arbitrary user in $\targetrel(r_i)$.
	\item If $\Gamma(r_i) \neq \emptyset$, then, for every $j \in \Gamma(r_i)$, define $\pi(s^i_j)$ as an arbitrary user in $\targetrel(r_i) \cap \targetrel(r_j)$, and set also $\pi(s^i_j) = \pi(s^j_i)$. 
\end{itemize}
	One can observe that $\pi$ is authorized and complete. By construction, every constraint $\left(s^i_j, s^j_i, =\right) \in C'$ is satisfied. Finally, for every $s, s' \in S'$ such that $(s, s', \neq)$, it must be the case that $s \in S^i$ and $s' \in S^j$ such that $(r_i, r_j, \updownarrow, \forall) \in C$, which implies that $\targetrel(r_i) \cap \targetrel(r_j) = \emptyset$. Hence we must have $\pi(s) \neq \pi(s')$, and $\pi$ is a valid plan.
\end{proof}

\subsection{Complexity of $\maxAPEParg{ \sodu }$}

We now introduce a particular version of \OAPEP, which seeks to find a valid authorization relation of maximum cardinality.
We write $\targetcard$ to denote the cardinality of such a relation.
Such a relation is, in some sense, a most resilient authorization relation possible, given the authorization constraints.
We call this problem \maxAPEP.
(We may also define a decision version \APEP to find resilient authorization relations.
We may, for example, introduce a global constraint $(\geqslant,t)$, which requires that at least $t$ users are authorized for each resource.
These types of problems are related to notions of resiliency in workflow systems~\cite{WaLi10}.) 
% = \underset{A~\mbox{\scriptsize{valid}}}{\operatorname{max}} \card{A}$.

In this section, $\left(\baserel,C\right) $ is an $\APEParg{\sodu}$ instance. 
In Theorem~\ref{sodufnc}, we established that $\DAPEParg{ \sodu }$ could be reduced to $\DAPEParg{\sodu,\fnc}$. 
Let $\Pi$ denote the set of valid solutions to instance $\left(\baserel,C \cup \set{(=,1)}\right)$ (that is, functions $\pi : R \rightarrow U$).
Given a function $\pi \in \Pi$, we say $A \subseteq U \times R$  {\em contains} $\pi$ if and only if for every $ r \in R, (\pi(r),r) \in A$. 
Let $M_{\pi}$ denote the maximum size of a valid authorization relations containing $\pi$.
Theorem~\ref{sodufnc} established that any solution $\targetrel$ of $\APEParg{ \sodu}$ contains at least one function $\pi \in \Pi$.
We write $\targetcard$ to denote $\max\set{M_{\pi} : \pi \in \Pi}$.
% \underset{\pi \in \Pi}{\operatorname{max}}~ M_{\pi}$.

\subsubsection{Patterns}

% Given a function $\pi : R \rightarrow U$, $P(\pi) = \left\lbrace \pi^{-1}(u):\ u \in U, \pi^{-1}(u) \neq \emptyset\right\rbrace$ partitions $R$ into non-empty subsets~\cite{CrGuKa15,KaGaGu15}. 
A function $\pi : R \rightarrow U$ defines an equivalence relation $\sim_{\pi}$ on $R$, where $r \sim_{\pi} r'$ iff $\pi(r) = \pi(r')$.
The equivalence classes defined by this relation form a partition of $R$ which we call the \emph{pattern} associated with $\pi$ and denote it by $P(\pi)$.
We say two functions $\pi$ and $\pi'$ are \emph{equivalent} if $P(\pi) = P(\pi')$.
% The definition of patterns is based on an equivalence relation of resources: $r$ and $r'$ are equivalent if $\pi(r) = \pi(r')$. 
% A pattern represent a family of functions which all admit the same equivalence relation. 
% Hence, two functions $\pi$ and $\pi'$ are said to be {\em equivalent} if they have the same pattern. 
% Let $P$ be a pattern, i.e. a partition of $R$. 
For UI constraints and any two functions $\pi$ and $\pi'$ such that $P(\pi) = P(\pi')$, $\pi$ is eligible iff $\pi'$ is eligible.
% such that $P(\pi) = P(\pi')$. 
Hence, we will say $P$ is \textit{eligible} if and only if, there exists $\pi$ such that $P = P(\pi)$ and $\pi$ is eligible for $C$. 
Henceforth, we only consider eligible patterns. 
We write  $M_{P}$ to denote $\max\set{M_{\pi} : P(\pi) = P}$.
% \underset{P(\pi) = P}{\operatorname{max}}~ M_{\pi}$. 
There exists an eligible pattern $P$ such that $\targetcard = M_{P}$.

Let us suppose that we are able, given a pattern $P$, to construct a valid $A$, such that $\card{A} = M_P$, in FPT time $f(k)n^{O(1)}$. There are at most $\mathcal{B}_k$ eligible patterns, where $\mathcal{B}_k$ is the Bell number and $\mathcal{B}_k = O(2^{k \log k})$ \cite{BeTa10}.%
\footnote{All logarithms in this paper are of base 2.}
Then, $\maxAPEParg{ \sodu}$ would be FPT: exploring all the eligible patterns and applying the FPT algorithm to compute $M_P$ for each $P$ is executed in time $O^{*}(2^{k \log k}f(k))$. 
As a consequence, our objective now is to design a FPT algorithm to compute $A_P$ such that $\card{A_P} = M_P$.

\subsubsection{Exploring patterns to solve $\maxAPEParg{ \sodu}$}

\begin{lemma}\label{lem10}
Let $P = \set{T_1,T_2,\ldots,T_d}$ be a pattern. An authorization relation $A_{P}$, such that $\card{A_P} = M_P$, can be computed in FPT time $O^*(2^k)$.
\end{lemma}

\begin{proof}
Clearly $d\le n$. We extend $P$ into $P^{*}$ in order to have $\card{P^{*}} = n$. We set $P^{*} = \left\lbrace T_1,\ldots,T_d,\emptyset_1,\ldots,\emptyset_{n-d}\right\rbrace=
\left\lbrace T_1,\ldots,T_n\right\rbrace.$ We build a weighted bipartite graph $G_P = \left( P^{*} \cup U, E, \omega\right)$, where $(T_i,u) \in E$ if and only if $T_i \subseteq \baserel(u)$ and for any $i \in \left[ n-d\right]$ and $u \in U$, $(\emptyset_i,u) \in E$. Assign to $(T_i,u) \in E$ weight $\omega(T_i,u)$ which is the cardinality of the maximum independent set in $\baserel(u)$ containing $T_i$:
%We say that resources $r$ and $r'$ are {\em independent} if and only if $(r,r',\neq,\forall) \notin C$. For any $i \in \left[ n-d\right]$ and $u \in U$, we define $\omega(\emptyset_i,u)$ as the cardinality of the maximum %independent set of resources in $\baserel(u)$. In other words, calculating $\omega(\emptyset_i,u)$ means determining the largest set $X \subseteq \baserel(u)$ such that, for any $r,r' \in X$, $(r,r',\neq,\forall) \notin C$. %For any $i,j \in \left[ n-d\right]$, we have $\omega(\emptyset_i,u) = \omega(\emptyset_j,u)$ Now, to any $(T_i,u) \in E$, we assign  weight $\omega(T_i,u)$ which is the cardinality of the maximum independent set in $%\baserel(u)$ containing $T_i$:
\[
\omega(T_i,u) =  \max\limits_{\substack{\forall (r,r') \in X^2, ~(r,r',\updownarrow,\forall) \notin C \\  T_i \subseteq X \subseteq \baserel(u)}} \vert X \vert
\]
There are at most $n^2$ weights to compute. For any $e \in E$, calculation of every $\omega(e)$ can be performed in time $O(2^k)$ by enumerating all subsets of $\baserel(u)$. Thus the bipartite graph $G_P$ can be  built in time $O^{*}(2^k)$. We solve the \textsc{assignment problem} on $G_P$ and obtain a maximum weighted matching (MWM), $\mathcal{M}$, in polynomial time using the Hungarian algorithm~\cite{Ku10}. 

For every edge $e \in \mathcal{M}$ we compute an independent set $X_e$ as follows. For each $(T_i,u) \in \mathcal{M}$, choose a maximum independent set $X_{(T_i,u)}$ such that $T_i \subseteq X_{(T_i,u)} \subseteq \baserel(u)$ and, therefore, $\card{X_{(T_i,u)}} = \omega(T_i,u)$. We define the authorization relation $A_{\mathcal{M}}$ such that $A_{\mathcal{M}}(u) = X_{(T_i,u)}$.
%For every $(\emptyset_i,u) \in \mathcal{M}$, we compute $X_{(\emptyset_i,u)} \subseteq \baserel(u)$ which is independent and satisfies $\card{X_{\emptyset_i,u}} = \omega(\emptyset_i,u)$.  We define the %authorization relation $A_{\mathcal{M}}$ such that, for any $u \in U$, $A_{\mathcal{M}}(u) = X_{e_u}$ where $u$ is an end-vertex of the edge $e_u \in \mathcal{M}$.

For any $u \in U$, $A_{\mathcal{M}}(u) \subseteq \baserel(u)$. Furthermore, for any $u \in U$, $A_{\mathcal{M}}(u)$ contains resources which are pairwise independent, so $A_{\mathcal{M}}$ is valid. We define the function $\tilde{\pi}$ such that $\tilde{\pi}(r) = u$ if and only if $r \in T_i$ and $(T_i,u) \in \mathcal{M}$. $A_{\mathcal{M}}$ contains $\tilde{\pi}$ whose pattern is $P$. Thus, $\card{A_{\mathcal{M}}} \leq M_P$.

We know that there exists a valid function $\pi$ such that $M_P = M_{\pi}$ and $P(\pi) = P = P(\tilde{\pi})$. There exists a matching $\mathcal{M}'$ representing $\pi$ in $G_P$. If $\pi(T_i) = u$, then $(T_i,u) \in \mathcal{M}'$. If $\pi^{-1}(u)$ is empty, we associate $u$ with an arbitrary vertex $\emptyset_i$. Since $\card{A_{\mathcal{M}}}$ is equal to the weight of the MWM $\mathcal{M}$ of $G_P$, we have $M_{\pi} \leq \card{A_{\mathcal{M}}}$. Hence, $M_P = \card{A_{\mathcal{M}}}$.
\end{proof}

In Figure~\ref{fig:maxrasp}, we use a simple example to illustrate the matching process described in the proof of Lemma~\ref{lem10}. 
We consider the pattern $P = \set{\set{r_1,r_4}, \set{r_2}, \set{r_3}}$ and merge $\emptyset_1$ and $\emptyset_2$ into a single node to keep the bipartite graph readable.
The figure shows $G_P$ (derived from $\baserel$) and the resulting MWM (where the matching is indicated by the thick lines). %To keep the bipartite graph readable, we. 
Then, for example, $\omega(\set{r_3}, u_3) = 2$ because $X = \set{r_1, r_3}$ is the largest independent subset of $\baserel(u_3)$ containing $r_3$.
 
\begin{figure}[h]
\centering

\begin{subfigure}[b]{0.48\columnwidth}\centering
 \begin{tikzpicture}
  \node[scale=0.8,draw,circle] (r1') at (2.0,5.7) {$r_1$};
  \node[scale=0.8,draw,circle] (r2') at (3.5,5.7) {$r_2$};
  \node[scale=0.8,draw,circle] (r3') at (3.5,4.5) {$r_3$};
  \node[scale=0.8,draw,circle] (r4') at (2.0,4.5) {$r_4$};
  \draw (r1') -- (r2');
  \draw (r3') -- (r2');
  \draw (r3') -- (r4');
 \end{tikzpicture}
\caption{Constraints}
\end{subfigure}
\hfill
\begin{subfigure}[b]{0.48\columnwidth}\centering
\[
\begin{array}{c|cccc}
& r_1 & r_2 & r_3 & r_4 \\
\hline
u_1 & & & & 1\\
u_2 & 1 & & & 1\\
u_3 & 1 & 1 & 1 & \\
u_4 & 1 & 1 & & \\
u_5 & & & & 1\\ 
\end{array}
\]
\caption{Base relation $\baserel$}
\end{subfigure}

\vspace*{\baselineskip}

\begin{subfigure}[b]{.48\columnwidth}\centering
 \begin{tikzpicture}[scale=.8,transform shape]
 %bipartite Graph

  \draw[fill=black] (5.7,0.9) circle (0.05);
  \draw[fill=black] (5.7,2.3) circle (0.05);
  \draw[fill=black] (5.7,3.7) circle (0.05);
  \draw[fill=black] (5.7,5.1) circle (0.05);
  \node at (5.0,5.1) {$\left\lbrace r_1,r_4 \right\rbrace $};
  \node at (5.2,3.7) {$\left\lbrace r_2 \right\rbrace $};
  \node at (5.2,2.3) {$\left\lbrace r_3 \right\rbrace $};
  \node at (5.1,0.9) {$\emptyset_1,\emptyset_2$};

  \draw[fill=black] (8.2,0.5) circle (0.05);
  \draw[fill=black] (8.2,1.8) circle (0.05);
  \draw[fill=black] (8.2,3.1) circle (0.05);
  \draw[fill=black] (8.2,4.4) circle (0.05);
  \draw[fill=black] (8.2,5.7) circle (0.05);
  \node at (8.6,5.7) {$u_1$};
  \node at (8.6,4.4) {$u_2$};
  \node at (8.6,3.1) {$u_3$};
  \node at (8.6,1.8) {$u_4$};
  \node at (8.6,0.5) {$u_5$};

  \draw (5.7,5.1) -- (8.2,4.4);
  \draw (5.7,3.7) -- (8.2,1.8);
  \draw (5.7,3.7) -- (8.2,3.1);
  \draw (5.7,2.3) -- (8.2,3.1);
  \draw (5.7,0.9) -- (8.2,5.7);
  \draw (5.7,0.9) -- (8.2,4.4);
  \draw (5.7,0.9) -- (8.2,3.1);
  \draw (5.7,0.9) -- (8.2,1.8);
  \draw (5.7,0.9) -- (8.2,0.5);
 \end{tikzpicture}
\caption{$G_P$}
\end{subfigure}
\begin{subfigure}[b]{.48\columnwidth}\centering
 \begin{tikzpicture}[scale=.8,transform shape]
 %bipartite Graph

  \draw[fill=black] (5.7,0.9) circle (0.05);
  \draw[fill=black] (5.7,2.3) circle (0.05);
  \draw[fill=black] (5.7,3.7) circle (0.05);
  \draw[fill=black] (5.7,5.1) circle (0.05);
  \node at (5.0,5.1) {$\left\lbrace r_1,r_4 \right\rbrace $};
  \node at (5.2,3.7) {$\left\lbrace r_2 \right\rbrace $};
  \node at (5.2,2.3) {$\left\lbrace r_3 \right\rbrace $};
  \node at (5.1,0.9) {$\emptyset_1,\emptyset_2$};

  \draw[fill=black] (8.2,0.5) circle (0.05);
  \draw[fill=black] (8.2,1.8) circle (0.05);
  \draw[fill=black] (8.2,3.1) circle (0.05);
  \draw[fill=black] (8.2,4.4) circle (0.05);
  \draw[fill=black] (8.2,5.7) circle (0.05);
  \node at (8.6,5.7) {$u_1$};
  \node at (8.6,4.4) {$u_2$};
  \node at (8.6,3.1) {$u_3$};
  \node at (8.6,1.8) {$u_4$};
  \node at (8.6,0.5) {$u_5$};

  \draw[thick] (5.7,5.1) -- (8.2,4.4);
  \draw[thick] (5.7,3.7) -- (8.2,1.8);
  \draw[color=black!20] (5.7,3.7) -- (8.2,3.1);
  \draw[thick] (5.7,2.3) -- (8.2,3.1);
  \draw[thick] (5.7,0.9) -- (8.2,5.7);
  \draw[color=black!20] (5.7,0.9) -- (8.2,4.4);
  \draw[color=black!20] (5.7,0.9) -- (8.2,3.1);
  \draw[color=black!20] (5.7,0.9) -- (8.2,1.8);
  \draw[thick] (5.7,0.9) -- (8.2,0.5);
 \end{tikzpicture}
\caption{MWM}
\end{subfigure}
\caption{Computing a maximum weighted matching for an instance of $\maxAPEParg{ \sodu }$}
\label{fig:maxrasp}
\end{figure}
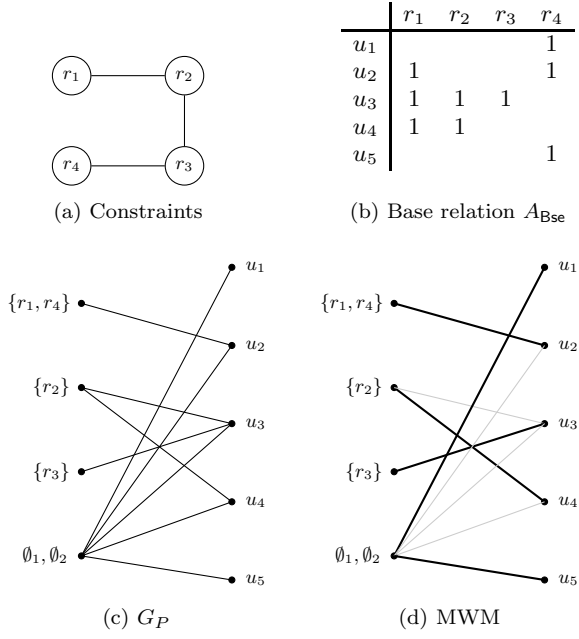

\begin{theorem}
$\maxAPEParg{ \sodu }$ (and thus  $\DAPEParg{ \sodu}$) can be solved in FPT time $O^*(2^{k+k \log k})$.
\end{theorem}

\begin{proof}
We explore all eligible patterns $P$. For each one, we construct $A_{P}$ using Lemma \ref{lem10}, and keep the largest one. The time complexity of this algorithm is $O^{*}\left( 2^{k + k \log k}\right)$. Hence, it is FPT parameterized by $k$.
\end{proof}

\subsection{Complexity of $\DAPEParg{ \sode }$}

In this section, we solve the $\maxAPEParg{ \sode}$ problem by reducing to the \textsc{max weighted partition} problem~\cite{BjHuKo09}: that is, given a ground set $K$ and $p$ functions $f_1,\ldots,f_p$ from $2^K$ to integers from the range $[-M,M]$, $M\ge 1$, find a partition $\left\lbrace K_1,\ldots,K_p \right\rbrace $ of $K$ that maximizes $\sum_{i=1}^{p} f_i(K_i)$.
The following result is a corollary of the main theorem on \textsc{max weighted partition} in \cite{BjHuKo09}.

\begin{lemma}\label{mwp-lem}
\textsc{max weighted partition} can be solved in time $O^*(2^k p^2 M).$ 
\end{lemma}

\begin{theorem}\label{maxsode-lem}
\sloppy
$\DAPEParg{ \sode}$ and $\maxAPEParg{ \sode }$ can be solved in time $\mathcal{O}^*(2^k)$.
\end{theorem}
\begin{proof}
We reduce to the \textsc{Max Weighted Partition} problem. The ground set is $R$, and we construct weight functions indexed by sets, which are elements of a family $\chi$ of subsets of $U$. We say that two resources $r, r' \in R$ are \emph{independent} if and only if $(r, r', \updownarrow, \exists) \notin C$. Moreover, the \emph{degree} $d(r)$ of a resource $r \in R$ is the number of resources $r'$ such that $(r, r', \updownarrow, \exists) \in C$. 
The index set $\chi$ is defined as $\bigcup_{r \in R} \chi_r$, where $\chi_r$ consists of all subsets of $\baserel(r)$ if $|\baserel(r)| \le \log k$, or $d(r)+1$ largest subsets of $\baserel(r)$ if $|\baserel| > \log k$ (breaking ties arbitrarily).
For every $X \in \chi$, we define a weight function $f_X : 2^R \rightarrow [-|\baserel|-1, |\baserel|]$ as follows: for every $T \subseteq R$, $f_X(T)$ is set to $|T||X|$ if $T$ is an independent set and $X \subseteq \baserel(r)$ for every $r \in T$, and $f_X(T) = -|\baserel|-1$ otherwise.
Now, we show that any valid authorization relation $\targetrel$ of $\DAPEParg{ \sode }$ corresponds to a partition of $R$ of cost $|\targetrel|$, and vice versa.

Let $\mathcal{T} = \{T_X: X \in \chi\}$ be a partition of $R$ of nonnegative weight. Note that if $f_X(T_X) \ge 0$ for every $X \in \chi$ then $\sum_{X\in \chi}f_X(T_X)\le |\baserel|$. Thus, since $\mathcal{T}$ is of nonnegative weight, $f_X(T_X) \ge 0$ for every $X \in \chi$.
%First observe that it implies that $f_X(T_X) \ge 0$ for every $X \in \chi$. 
Construct an authorization relation $\targetrel$ such that for any $X \in \chi$, for any $r \in T_X$, we have $\targetrel(r) = X$. Obviously, since $\mathcal{T}$ is a partition of $R$, $\targetrel$ is complete. Then, by definition of $\chi$, we have $\targetrel(r) \subseteq \baserel(r)$ and thus $\targetrel$ is authorized. Finally, for any $r, r' \in R$ such that $\targetrel(r) = \targetrel(r')$, it must hold that $r, r' \in T_{\targetrel(r)}$, and since $T_{\targetrel(r)}$ is independent, $(r, r', \updownarrow, \exists) \notin C$, and $\targetrel$ is eligible. In other words, $\targetrel$ is a valid authorization relation, and its weight is $\sum_{X \in \chi} f_X(T_X)$.

For any valid authorization relation $\targetrel$, let $P(\targetrel)$ be the partition of $R$ into equivalence classes with respect to the following equivalence relation: $r, r' \in R$ are equivalent if and only if $\targetrel(r) = \targetrel(r')$. We now prove that there always exists a valid authorization relation $\targetrel'$ of size at least $|\targetrel|$ such that $\targetrel'(r) \in \chi$ for every $r \in R$. If this is true, then it will mean that we may assume that $P(\targetrel) = \{T_X: X \in \chi\}$, and since $\targetrel$ is valid, $\sum_{X \in \chi} f_X(T_X) = |\targetrel|$.
For $r \in R$, if $|\baserel(r)| \le \log k$, then since $\chi$ contains all subsets of $\baserel(r)$, it holds that $\targetrel(r) \in \chi$. If $|\baserel(r)| > \log k$ and $\targetrel(r) \notin \chi_r$, recall that $\chi_r$ consists of $d(r)+1$ largest subsets of $\baserel(r)$. Hence, there must exist $X \in \chi_r$ such that $\targetrel(r') \neq X$ for every $r'$ such that $(r, r', \updownarrow, \exists) \in C$. Hence, replacing $\targetrel(r)$ by $X$ creates another valid authorization relation $\targetrel'$ of size at least $|\targetrel|$ and such that $\targetrel'(r) \in \chi_r$. Repeating this modification for every $r \in R$ such that $\targetrel(r) \notin \chi$, we end up with a valid authorization having the desired property.

Using the reduction above together with Lemma \ref{mwp-lem}, we prove the claimed statement.
\end{proof}

\section{Discussion}

\subsection{Constraint types}\label{sec:extended-constraint-types}

In Section~\ref{sec:bod-and-sod-constraints} we identified a number of constraint types of the form $(r,r',\circ,Q)$, where $r$ and $r'$ are resources, $\circ$ is a logical binary operator, and $Q$ is a quantifier.
For ease of reference we summarize these constraints and the respective conditions for satisfaction in Table~\ref{tbl:constraint-summary}.

\begin{table}[h]\setlength{\extrarowheight}{2pt}\centering
 \begin{tabular}{|c|c|}
 \hline
  $(r,r',\leftrightarrow,\forall)$ & $A(r) = A(r')$ \\
  $(r,r',\leftrightarrow,\exists)$ & $A(r) \cap A(r') \ne \emptyset$ \\
  $(r,r',\updownarrow,\forall)$ & $A(r) \cap A(r') = \emptyset$ \\
  $(r,r',\updownarrow,\exists)$ & $A(r) \ne A(r')$ \\
  $(r,r',\rightarrow,\forall)$ & $A(r) \subseteq A(r')$ \\
 \hline
 \end{tabular}
 \caption{Constraint types defined in Section~\ref{sec:bod-and-sod-constraints}}\label{tbl:constraint-summary}
\end{table}

We chose to introduce the constraints in Table~\ref{tbl:constraint-summary} because of their obvious connections to known constraints in the literature and to simplify the exposition of the technical material.
We now discuss ways in which these constraints could be extended.
Notice that the satisfaction of each constraint may be defined in terms of $A(r)$ and $A(r')$.

One obvious extension, then, is to define constraints of the form $((r_1,\dots,r_m),\circ,Q)$, and to define constraint satisfaction in terms of $A(r_i)$, $1 \leqslant i \leqslant m$.
We may define constraint satisfaction in a number of ways, including (but not limited to) the following: %
\begin{inparaenum}[(i)]
 \item for \emph{all} $i$ and $j$, $1 \leqslant i < j \leqslant m$, $(r_i,r_j,\circ,Q)$ is satisfied; or
 \item for \emph{some} $i$ and $j$, $1 \leqslant i < j \leqslant m$, $(r_i,r_j,\circ,Q)$ is satisfied.
%  \item $\bigcap_{i=1}^m A(r_i) = \emptyset$.
\end{inparaenum}
Note, however, that the first of these choices can be realized simply by defining a set of constraints $\set{(r_i,r_j,\circ,Q) : 1 \leqslant i < j \leqslant m}$.
% Hence, we will assume the second choice is used.
% {\color{red}Are these constraints $f(k)$-bounded?}

\sloppy Consider the constraint $((r_1,\dots,r_m),\updownarrow,\forall)$, and suppose, as another alternative for constraint satisfaction, we require that $\bigcap_{i=1}^m A(r_i) = \emptyset$.
% which requires $A(r_i) \cap A(r_j) = \emptyset$ for some $i$ and $j$.
In other words, there is no user that is assigned to all resources in the set $\set{r_1,\dots,r_m}$.
It is easy to see that such a constraint is $k$-bounded (since removing a user from a valid relation can only affect completeness, not the eligibility, of the relation).
Thus, with this interpretation, $((r_1,\dots,r_m),\updownarrow,\forall)$ represents a canonical SMER constraint~\cite{LiTrBi07} (if the set of resources is interpreted as a set of mutually exclusive roles).
We return to SMER constraints in Section~\ref{sec:resiliency-and-sod}.

Another possible extension is to define constraints of the form $(R',R'',\circ,Q)$ and to define constraint satisfaction in terms of $A(R')$ and $A(R'')$.
For example, the constraint $(R',R'',\updownarrow,\forall)$ requires that $A(R') \cap A(R'') = \emptyset$.
Again, constraints of this form are $k$-bounded.
In other words, the users assigned to resources in $R'$ are different from the users assigned to resources in $R''$.
This constraint, therefore, allows us to specify that resources should be allocated to disjoint \emph{teams} of users (rather than just individual users).
Of course such constraints could be nested: we might define a further constraint $(R''_1,R''_2,\updownarrow,\forall)$ where $R''_1$ and $R''_2$ are subsets of $R''$.

\subsection{Resiliency in access control}

Suppose we are given an authorization relation $A \subseteq U \times R$ and a set of resources $Q \subseteq R$.
Then a resiliency policy is defined by a tuple $(Q,s,d,t)$, where $s$, $d$ and $t$ are integers~\cite{LiWaTr09}.
The policy is satisfied if, following the removal of any $s$ users from $U$, there exist $d$ disjoint teams of users, $U_1,\dots,U_d$, such that $A(U_i) \supseteq Q$ and $\card{U_i} \leqslant t$ for each $i$.

The \emph{resiliency checking problem}~\cite{LiWaTr09} asks whether a resiliency policy is satisfiable or not.
It has been shown that the hard part of the problem is finding the teams (since we can enumerate all possible user sets that are missing $s$ users), so research has focused on solving the problem for instances in which $s = 0$~\cite{CrGuWa15,LiWaTr09}.

Informally, a solution of the resiliency checking problem may be viewed as a function mapping (different copies of the set of) resources to users.
Thus we can transform an instance of the resiliency checking problem (where $s = 0$) into an instance of \DAPEP.
We define $d$ copies of each resource in $Q$; we write $r^{(i)}$ to denote the $i$th copy of resource $r$ in $Q$.
We then define \[ \baserel = \set{(u,r^{(i)}) : (u,r) \in A, 1 \leqslant i \leqslant d}. \]
Finally, we define the global constraint $(=,1)$ and, for all $r_1,r_2 \in Q$ and all $i$ and $j$ such that $1 \leqslant i < j \leqslant d$, we define a constraint $(r_1^{(i)},r_2^{(j)},\updownarrow,\forall)$.
The authorization relation $\baserel$ ensures that each user is authorized according to the original relation $A$.
The global constraint ensures that each resource is assigned to a single user.
The other constraints ensure that a user is only assigned to resources in one copy of $Q$.

The results in Section~\ref{sec:complexity-fk-bounded} assert that the resulting problem is FPT.
Hence, the resiliency checking problem is also FPT (confirming an earlier result of Crampton \emph{et al.}~\cite{CrGuWa15}).

Note, finally, that we can simplify the above construction, using the constraints introduced in Section~\ref{sec:extended-constraint-types}: we use $Q^{(i)}$ to denote the $i$th copy of the set of resources $Q$ and define the set of constraints
\[
 \set{\big(Q^{(i)},Q^{(j)},\updownarrow,\forall\big) : 1 \leqslant i < j \leqslant d}.
\]

\subsection{Resiliency and separation of duty}\label{sec:resiliency-and-sod}

The RBAC96 standard discusses constraints based on mutually exclusive roles~\cite{SaCoFeYo96}.
Such a constraint is defined by a set of roles $R_{\sf mutex}$ and is satisfied by the user-role assignment relation provided no user is assigned to more than one role in $R_{\sf mutex}$.

Li, Tripunitara and Bizri introduced the more general \emph{static mutually exclusive role} (SMER) constraints~\cite{LiTrBi07}, which have the form $(R_{\sf mutex},t)$, where $t \leqslant \card{R_{\sf mutex}}$.
(A canonical SMER constraint has $t = \card{R_{\sf mutex}}$.)
Such a constraint is satisfied provided every user is assigned fewer than $t$ of the roles in $R_{\sf mutex}$.
We can check whether a user-role assignment relation satisfies a SMER constraint in polynomial time~\cite{LiTrBi07}, informally because we only need to consider each user once.

Li \emph{et al.} went on to distinguish SMER constraints from \emph{static separation of duty} (SSoD) policies~\cite{LiTrBi07}, which are defined by a set of permissions $P$ and an integer $t \leqslant \card{P}$.
Such a constraint is satisfied if no subset of fewer than $t$ users is collectively authorized for the permissions in $P$.
Checking whether a SSoD policy is satisfied by a given user-role assignment relation is computationally hard, informally because we need to consider every possible subset of users having cardinality less than $t$.

\sloppy
Li, Wang and Tripunitara studied the complexity of determining whether it was possible to simultaneously satisfy static separation of duty constraints and a resiliency policy~\cite{LiWaTr09}.
Unsurprisingly, it is computationally hard to decide this question, given that it is hard to decide whether an authorization relation satisfies a static separation of duty policy~\cite{LiTrBi07}.
However, they did not consider the possibility of simultaneously satisfying SMER constraints and resiliency policies.
Now observe that a SMER constraint is user-independent and is $k$-bounded.
Thus, for example, it is possible to develop an FPT algorithm to determine whether there exists an authorization relation $A \subseteq U \times R$ such that a resiliency policy and a set of SMER constraints are simultaneously satisfied.

\section{Conclusion}\label{sec:conclusion}

In this paper we have introduced a more general framework for articulating problems of finding authorization relations (``policies'') that must satisfy certain kinds of constraints.
We have shown that there exist FPT algorithms to solve  the authorization policy existence problem when all constraints are user-independent and are bounded in an appropriate way.
We have also shown that many constraints of practical interest are indeed user-independent and bounded.

We have chosen to consider user-independent constraints, not least because such constraints have been studied extensively in the literature on workflow satisfiability.
In fact, we could equally well consider resource-independent constraints because our framework is symmetric in a way that workflow satisfiability questions are not.
So, for example, we could define a constraint of the form $(u,u',\updownarrow,\forall)$ which would be satisfied provided the set of resources assigned to $u$ is distinct to the set of resources assigned to $u'$.
In this way, we search for authorization relations that guarantee certain users do not have access to the same resources.
Moreover, if the number of users is small relative to the number of resources, which may well be the case in some multi-user systems (such as file systems), then $n$ will be the small parameter and the symmetry of our framework admits FPT algorithms for solving problem instances of this form.

We believe there are many opportunities for future work, not least exploring what types of authorization constraints might be useful in practice and determining whether those constraints are user-independent and bounded.

%\clearpage 

\appendix\section{Results for bounded constraints}

\begin{proposition}\label{prop:bindingbcc}
% \sloppy 
Constraints $(r',r'',\rightarrow,\forall)$, $\left(r',r'',\leftrightarrow,\forall\right)$ and $\left(r',r'',\leftrightarrow,\exists\right)$ are $(k-1)$-bounded.
\label{th-boundedc-bodu}
\end{proposition}

\begin{proof}
\sloppy
Let $c$ be one of $(r',r'',\rightarrow,\forall)$, $\left(r',r'',\leftrightarrow,\forall\right)$ or $\left(r',r'',\leftrightarrow,\exists\right)$.
% Let $S=\{u^*\}\cup \{u_r:\ r \in R\setminus \left\{r',r''\right\}\}$ is the set of $k-1$ users (i.e., all users in $S$ are distinct).
Let $V$ be a set of $k-1$ distinct users and consider $A$, where $A(r') = A(r'') = \set{u}$ for some $u \in V$ and, for any $r_1,r_2 \in R\setminus \left\{r',r''\right\}$, $\card{A(r_1)} = 1$, $A(r_1) \ne A(r')$, and $A(r_1) \ne A(r_2)$.
% Then $\psi_{c,A}(\I) = V$ and one can check that $\Psi_c(\I) \le \psi_{c, A}(\I)$: indeed, for any solution $A'$, any subset of $A'(R)$ of size at least $k$ must contain two users who are both assigned to the same resource, and thus one of them can be removed without affecting completeness or satisfiability. We thus conclude $\Psi_c(\I) \le k-1$.
Then $\ker(A : U \times R,c) = V$ and $\card{\ker(A : U \times R,c)} = k-1$.
Moreover, for any relation $A'$ valid with respect to $U \times R$ and $c$, any subset of $A'(R)$ of size at least $k$ must contain two users who are both assigned to the same resource; thus one of them can be removed without affecting completeness or satisfiability.
Hence, the constraint is $(k-1)$-bounded.
\end{proof}

\begin{proposition}
Constraints $\left(r',r'',\updownarrow,\forall\right)$ and $\left(r',r'',\updownarrow,\exists\right)$ are $k$-bounded.
\label{th-boundedc-sodu}
\end{proposition}

\begin{proof}
Let $V$ be a set of $k$ distinct users and consider $A$ where $\card{A(r)}=1$ for each $r\in R$ and $A(R)=V$.
Then $\ker(A : U \times R,c) = V$ and $\card{\ker(A : U \times R,c)} = k$.
% $\psi_{c,A}(\I) = V$. 
Now, for any for any relation $A'$ valid with respect to $U \times R$ and $c$, any subset of $A'(R)$ of size at least $k+1$ must contain two users who are both assigned to the same resource, and thus one of them can be removed without violating completeness or satisfiability. 
Thus $\card{\ker(A' : U \times R , c)} \leqslant \card{\ker(A : U \times R, c)}$, from which the result follows.
% This proves that $\psi_{c, A'}(\I) \le \psi_{c, A}(\I)$ for any solution $A'$, and thus $\Psi_c(\I) \le k$.
% Value $\card{\psi_{c,A}}$ is maximal when, for any $r \in R$, $A(r) = \left\{u_r\right\}$ and, for any $r_1 \neq r_2$, we have $u_{r_1} \neq u_{r_2}$. Thus $\Psi_{c} = k$.
\end{proof}

\begin{proposition}
Constraint $(R',\leq,t)$ is $k$-bounded.
\end{proposition}

\begin{proof}
\sloppy
Given any valid solution $A$, the removal of any user cannot make $A$ non-eligible with respect to \mbox{$c=(R', \leq, t)$}, but may violate completeness. 
Hence, $\ker(A : U \times R,c)$ is largest when $|A(r)|=1$ for all $r$ and $A(R) = k$, in which case we have \mbox{$\card{\ker(A : U \times R,c)} = k$.}
%Let constraint $c$ be $(R',\leq,t)$. We partition $U$ into two subsets $U'$ and $U''$ such that $\card{U'} \leq t$ and we set $A(R')=U'$ and $A(R\backslash R') = U''$. We suppose $A$ is complete. Removing any user in $U''$ does not affect the eligibility, but only the completeness of $A$. This is why we can remove at most $\card{R \backslash R'}$ users. Furthermore, removing any user in $U'$ does not affect the eligibility too because, in any case, the cardinality of $A(R')$, after removal, remains smaller than $t$. So, we can remove at most $\card{R'}$ users if $t \geq \card{R'}$ and $t$ users if $t < \card{R'}$. Hence, we can remove at most $\card{R'} + \card{R\backslash R'} = k$ users.
\end{proof}

Similarly, the global cardinality constraint $c = (\le,t)$ is $k$-bounded because, in this case too, any removal does not affect the eligibility of the relation; it can only affect the completeness.
% Thus, for any satisfiable instance $\I$, $\Psi_{c}(\I) \leq k$. 
Obviously, these results remain true with $<$ instead of $\le$. However, as we will see, they do not hold if we replace $\le$ by $=$ or $\ge$. 
Indeed, a constraint such as $(=, t)$ requires that some set of $t$ users cannot be removed. 
Hence, if $t$ is not bounded by a function of $k$ only, the constraint is not $f(k)$-bounded for any computable function $f$.

\begin{proposition}\label{prop:cardinalitybcc}
Constraints $(R',=,t)$ and $(R',\geq,t)$ are $2 \max\{k, t\}$-bounded, but not $(\max\{k, t\}-1)$-bounded.
\end{proposition}

\begin{proof}
We only give the proof for $(R', =, t)$, the other one being similar. 
One can observe that $\ker(A : U \times R,c)$ is largest when $|A(r)|=1$ for any $r \in R \setminus R'$, $|A(R \setminus R')| = |R \setminus R'|$, and $A(R') \cap A(R \setminus R') = \emptyset$.
In this case we have $\card{\ker(A : U \times R,c)} \le |R \setminus R'| + t \le 2 \max\{k, t\}$. 
% maximum value of $\psi_{c, A}(\I)$ is reached when $|A(r)|=1$ for any $r \in R \setminus R'$, $|A(R \setminus R')| = |R \setminus R'|$, and $A(R') \cap A(R \setminus R') = \emptyset$, in which case we have $\psi_{c, A}(\I) \le |R \setminus R'| + t \le 2 \max\{k, t\}$.

Concerning the negative result, observe that if $\max\{k, t\} = t$, then no user of $A(R')$ can be removed from any valid solution, and if $\max\{k, t\} =k$, then there exists solutions in which $A(R) \ge k$ and the removal of any user from $A(R)$ either violates a constraint or breaks completeness. 
%Let $c$ be either $\left(R',=,t\right)$ or $\left(R',\geq,t\right)$. Let $U'$ be a subset of $U$ such that $\card{U'} = t$. For any $r' \in R'$, we set $A(r')= U'$ and we suppose $A$ is complete. It is impossible to remove any user in $U'$ because the relation would become ineligible: for any $u \in U'$, $\card{A[-u](R')} < t$. Hence, $\Psi_c(\I) \geq t$. If we fix $k$ and assume that $t$ is a monotonically increasing and unbounded function $g(n)$ such that $g(n)\le n$. Then, $\Psi_c(\I)$ tends to infinity when $n \rightarrow \infty$.
\end{proof}


\begin{thebibliography}{10}

\bibitem{BaBuKa14}
{\sc Basin, D.~A., Burri, S.~J., and Karjoth, G.}
\newblock Obstruction-free authorization enforcement: Aligning security and
  business objectives.
\newblock {\em Journal of Computer Security 22}, 5 (2014), 661--698.

\bibitem{BeTa10}
{\sc Berend, D., and Tassa, T.}
\newblock Improved bounds on {B}ell numbers and on moments of sums of random
  variables.
\newblock {\em Probability and Math. Statistics 30}, 2 (2010), 185--205.

\bibitem{BeFeAt99}
{\sc Bertino, E., Ferrari, E., and Atluri, V.}
\newblock The specification and enforcement of authorization constraints in
  workflow management systems.
\newblock {\em ACM Trans. Inf. Syst. Secur. 2}, 1 (1999), 65--104.

\bibitem{BjHuKo09}
{\sc Bj{\"{o}}rklund, A., Husfeldt, T., and Koivisto, M.}
\newblock Set partitioning via inclusion-exclusion.
\newblock {\em {SIAM} J. Comput. 39}, 2 (2009), 546--563.

\bibitem{BrNa89}
{\sc Brewer, D. F.~C., and Nash, M.~J.}
\newblock The {C}hinese wall security policy.
\newblock In {\em Proceedings of the 1989 {IEEE} Symposium on Security and
  Privacy\/} (1989), pp.~206--214.

\bibitem{CoCrGaGuJo14}
{\sc Cohen, D., Crampton, J., Gagarin, A., Gutin, G., and Jones, M.}
\newblock Iterative plan construction for the workflow satisfiability problem.
\newblock {\em J. Artif. Intell. Res. {(JAIR)} 51\/} (2014), 555--577.

\bibitem{Cr05}
{\sc Crampton, J.}
\newblock A reference monitor for workflow systems with constrained task
  execution.
\newblock In {\em SACMAT\/} (2005), E.~Ferrari and G.-J. Ahn, Eds., ACM,
  pp.~38--47.

\bibitem{CrGaGuJo15}
{\sc Crampton, J., Gagarin, A.~V., Gutin, G., and Jones, M.}
\newblock On the workflow satisfiability problem with class-independent
  constraints.
\newblock In {\em 10th International Symposium on Parameterized and Exact
  Computation, {IPEC} 2015, September 16-18, 2015, Patras, Greece\/} (2015),
  T.~Husfeldt and I.~A. Kanj, Eds., vol.~43 of {\em LIPIcs}, Schloss Dagstuhl -
  Leibniz-Zentrum fuer Informatik, pp.~66--77.

\bibitem{CrGuKa15}
{\sc Crampton, J., Gutin, G., and Karapetyan, D.}
\newblock Valued workflow satisfiability problem.
\newblock In {\em Proceedings of the 20th {ACM} Symposium on Access Control
  Models and Technologies\/} (2015), pp.~3--13.

\bibitem{CrGuWa15}
{\sc Crampton, J., Gutin, G., and Watrigant, R.}
\newblock Resiliency policies in access control revisited.
\newblock In {\em Proceedings of the 21st {ACM} on Symposium on Access Control
  Models and Technologies\/} (2016), {ACM}, pp.~101--111.

\bibitem{CrGuYe13}
{\sc Crampton, J., Gutin, G., and Yeo, A.}
\newblock On the parameterized complexity and kernelization of the workflow
  satisfiability problem.
\newblock {\em {ACM} Trans. Inf. Syst. Secur. 16}, 1 (2013), 4.

\bibitem{DowFel13}
{\sc Downey, R.~G., and Fellows, M.~R.}
\newblock {\em Fundamentals of Parameterized Complexity}.
\newblock Springer Verlag, 2013.

\bibitem{GlGaFe98}
{\sc Gligor, V.~D., Gavrila, S.~I., and Ferraiolo, D.~F.}
\newblock On the formal definition of separation-of-duty policies and their
  composition.
\newblock In {\em Security and Privacy - 1998 {IEEE} Symposium on Security and
  Privacy, Proceedings\/} (1998), {IEEE} Computer Society, pp.~172--183.

\bibitem{KaGaGu15}
{\sc Karapetyan, D., Gagarin, A., and Gutin, G.}
\newblock Pattern backtracking algorithm for the workflow satisfiability
  problem with user-independent constraints.
\newblock In {\em Frontiers in Algorithmics - 9th International Workshop,
  Proceedings\/} (2015), pp.~138--149.

\bibitem{Ku10}
{\sc Kuhn, H.~W.}
\newblock The {H}ungarian method for the assignment problem.
\newblock In {\em 50 Years of Integer Programming 1958-2008 - From the Early
  Years to the State-of-the-Art}. 2010, pp.~29--47.

\bibitem{LiTrBi07}
{\sc Li, N., Tripunitara, M.~V., and Bizri, Z.}
\newblock On mutually exclusive roles and separation-of-duty.
\newblock {\em {ACM} Trans. Inf. Syst. Secur. 10}, 2 (2007).

\bibitem{LiWaTr09}
{\sc Li, N., Wang, Q., and Tripunitara, M.~V.}
\newblock Resiliency policies in access control.
\newblock {\em {ACM} Trans. Inf. Syst. Secur. 12}, 4 (2009).

\bibitem{MaMoMo14}
{\sc Mace, J.~C., Morisset, C., and van Moorsel, A. P.~A.}
\newblock Quantitative workflow resiliency.
\newblock In {\em Computer Security - {ESORICS} 2014 - 19th European Symposium
  on Research in Computer Security, Wroclaw, Poland, September 7-11, 2014.
  Proceedings, Part {I}\/} (2014), M.~Kutylowski and J.~Vaidya, Eds., vol.~8712
  of {\em Lecture Notes in Computer Science}, Springer, pp.~344--361.

\bibitem{Ni06}
{\sc Niedermeier, R.}
\newblock {\em Invitation to fixed-parameter algorithms}.
\newblock Oxford University Press, 2006.

\bibitem{RoSuMaVaAt15}
{\sc Roy, A., Sural, S., Majumdar, A.~K., Vaidya, J., and Atluri, V.}
\newblock Minimizing organizational user requirement while meeting security
  constraints.
\newblock {\em {ACM} Trans. Management Inf. Syst. 6}, 3 (2015), 12.

\bibitem{SaCoFeYo96}
{\sc Sandhu, R.~S., Coyne, E.~J., Feinstein, H.~L., and Youman, C.~E.}
\newblock Role-based access control models.
\newblock {\em {IEEE} Computer 29}, 2 (1996), 38--47.

\bibitem{ScLoSo06}
{\sc Schaad, A., Lotz, V., and Sohr, K.}
\newblock A model-checking approach to analysing organisational controls in a
  loan origination process.
\newblock In {\em {SACMAT} 2006,11th {ACM} Symposium on Access Control Models
  and Technologies, Lake Tahoe, California, USA, June 7-9, 2006, Proceedings\/}
  (2006), D.~F. Ferraiolo and I.~Ray, Eds., {ACM}, pp.~139--149.

\bibitem{SiZu97}
{\sc Simon, R.~T., and Zurko, M.~E.}
\newblock Separation of duty in role-based environments.
\newblock In {\em 10th Computer Security Foundations Workshop {(CSFW} '97),
  June 10-12, 1997, Rockport, Massachusetts, {USA}\/} (1997), {IEEE} Computer
  Society, pp.~183--194.

\bibitem{WaLi10}
{\sc Wang, Q., and Li, N.}
\newblock Satisfiability and resiliency in workflow authorization systems.
\newblock {\em {ACM} Trans. Inf. Syst. Secur. 13}, 4 (2010), 40.

\end{thebibliography}
\end{document}